\documentclass[hidelinks,runningheads]{llncs}

\usepackage[english]{babel}
\usepackage[T1]{fontenc}
\usepackage[utf8]{luainputenc}

\usepackage{algorithmicx}
\usepackage{algpseudocode}

\usepackage{amsmath}
\usepackage{amssymb}
\usepackage{mathtools}
\usepackage{thm-restate}

\usepackage[english]{varioref}

\usepackage{pifont}
\usepackage{todonotes}

\usepackage{multirow}
\usepackage{makecell}
\usepackage{colortbl} 
\usepackage{caption} 
\usepackage{tikz}
\usetikzlibrary{automata,positioning,arrows.meta,patterns}
\tikzset{every state/.style={minimum size=0pt}}

\usetikzlibrary{calc,intersections,through,backgrounds,matrix,patterns}
\usepackage{pgfplots}
\usepgfplotslibrary{fillbetween}
\pgfdeclarelayer{bg}
\pgfsetlayers{bg,main}

\tikzset{
    hatch distance/.store in=\hatchdistance,
    hatch distance=10pt,
    hatch thickness/.store in=\hatchthickness,
    hatch thickness=0.3pt
}

\makeatletter
\pgfdeclarepatternformonly{mydots}
{\pgfqpoint{-1pt}{-1pt}}
{\pgfqpoint{1pt}{1pt}}
{\pgfqpoint{10pt}{10pt}}{\pgfpathcircle{\pgfqpoint{0pt}{0pt}}{.5pt}\pgfusepath{fill}}

\pgfdeclarepatternformonly[\hatchdistance,\hatchthickness]{northeast}
{\pgfqpoint{0pt}{0pt}}
{\pgfqpoint{\hatchdistance}{\hatchdistance}}
{\pgfpoint{\hatchdistance-1pt}{\hatchdistance-1pt}}{
    \pgfsetcolor{\tikz@pattern@color}
    \pgfsetlinewidth{\hatchthickness}
    \pgfpathmoveto{\pgfqpoint{0pt}{0pt}}
    \pgfpathlineto{\pgfqpoint{\hatchdistance}{\hatchdistance}}
    \pgfusepath{stroke}
}

\pgfdeclarepatternformonly[\hatchdistance,\hatchthickness]{northwest}
{\pgfqpoint{0pt}{0pt}}
{\pgfqpoint{\hatchdistance}{\hatchdistance}}
{\pgfpoint{\hatchdistance-1pt}{\hatchdistance-1pt}}{
    \pgfsetcolor{\tikz@pattern@color}
    \pgfsetlinewidth{\hatchthickness}
    \pgfpathmoveto{\pgfqpoint{\hatchdistance}{0pt}}
    \pgfpathlineto{\pgfqpoint{0pt}{\hatchdistance}}
    \pgfusepath{stroke}
}

\pgfdeclarepatternformonly[\hatchdistance,\hatchthickness]{horizontal}
{\pgfqpoint{0pt}{0pt}}
{\pgfqpoint{\hatchdistance}{\hatchdistance}}
{\pgfpoint{\hatchdistance-1pt}{\hatchdistance-1pt}}{
    \pgfsetcolor{\tikz@pattern@color}
    \pgfsetlinewidth{\hatchthickness}
    \pgfpathmoveto{\pgfqpoint{0pt}{0pt}}
    \pgfpathlineto{\pgfqpoint{\hatchdistance}{0pt}}
    \pgfusepath{stroke}
}

\pgfdeclarepatternformonly[\hatchdistance,\hatchthickness]{vertical}
{\pgfqpoint{0pt}{0pt}}
{\pgfqpoint{\hatchdistance}{\hatchdistance}}
{\pgfpoint{\hatchdistance-1pt}{\hatchdistance-1pt}}{
    \pgfsetcolor{\tikz@pattern@color}
    \pgfsetlinewidth{\hatchthickness}
    \pgfpathmoveto{\pgfqpoint{0pt}{0pt}}
    \pgfpathlineto{\pgfqpoint{0pt}{\hatchdistance}}
    \pgfusepath{stroke}
}
\makeatother

\makeatletter
\@ifclassloaded{llncs}{\newcommand{\lncsqed}{\qed}
}{\newcommand{\lncsqed}{}
}
\makeatother

\usepackage{xargs}
\newcommandx{\mytodo}[2][1=]{\todo[linecolor=red,backgroundcolor=red!25,bordercolor=red,#1]{#2}}

\usepackage{xspace} \newcommand{\cmark}{\ding{51}}\newcommand{\xmark}{\ding{55}}

\newcommand{\op}[1]{\ensuremath{\mathsf{#1}}\xspace}

\newcommand{\mc}[1]{\mathcal{#1}}

\newcommand{\complclass}[1]{\op{#1}}

\newcommand{\mrm}[1]{\ensuremath{\mathrm{#1}}}
\newcommand{\mbb}[1]{\mathbb{#1}}
\newcommand{\EDA}{\mrm{EDA}\xspace}
\newcommand{\EDAF}{\mrm{EDA}_F\xspace}
\newcommand{\IDA}{\mrm{IDA}\xspace}
\newcommand{\IDAF}{\mrm{IDA}_F\xspace}

\newcommand{\DBA}{\op{DBA}}
\newcommand{\NBA}{\op{NBA}}

\newcommand{\DWA}{\op{DWA}}

\newcommand{\PBA}{\op{PBA}}
\newcommand{\PCoBA}{\op{PCA}}
\newcommand{\PWA}{\op{PWA}}

\newcommand{\HPBA}{\op{HPBA}}
\newcommand{\SPBA}{\op{SPBA}}

\newcommand{\scc}{\op{scc}}

\newcommand{\cupdot}{\mathbin{\mathaccent\cdot\cup}}

\newcommand{\underlyingnba}[1]{#1^{\lhd}}

\newcommand{\Nat}{\mathbb{N}}
 
\usepackage{graphicx}

\makeatletter
\@ifpackageloaded{hyperref}{}{
\RequirePackage[bookmarks,unicode,colorlinks=true]{hyperref}
}
\def\@citecolor{blue}\def\@urlcolor{blue}\def\@linkcolor{blue}
\def\orcidID#1{\smash{\href{http://orcid.org/#1}{\protect\raisebox{-1.25pt}{\protect\includegraphics{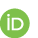}}}}}
\makeatother

\usepackage{marvosym} 
\begin{document}
\title{Ambiguity, Weakness, and Regularity\\ in Probabilistic Büchi Automata}
\author{Christof Löding\textsuperscript{(\Letter)} \and
  Anton Pirogov\textsuperscript{(\Letter)} \orcidID{0000-0002-5077-7497}
\thanks{This work is supported by the German research council (DFG) Research Training Group 2236 UnRAVeL}
\thanks{The final authenticated publication is available online at \texttt{https://doi.org/10.1007/978-3-030-45231-5\_27}}
}
\authorrunning{C. Löding and A. Pirogov}
\institute{RWTH Aachen University, Templergraben 55, 52062 Aachen, Germany \\
\email{\{loeding,pirogov\}@cs.rwth-aachen.de} }
\maketitle              \begin{abstract}
Probabilistic Büchi automata are a natural generalization of PFA to infinite words, but
have been studied in-depth only rather recently and many interesting questions are still
open. PBA are known to accept, in general, a class of languages that goes beyond the
regular languages. In this
work we extend the known classes of restricted PBA which are still regular,
strongly relying on notions concerning ambiguity in classical $\omega$-automata.
Furthermore, we investigate the expressivity of the not yet considered but natural class
of weak PBA, and we also show that the regularity problem for weak PBA is undecidable.
\end{abstract}

\section{Introduction}
\label{sec:intro}

Probabilistic finite automata (PFA) are defined similarly to
nondeterministic finite automata (NFA) with the difference that each
transition is equipped with a probability (a value between 0 and 1),
such that for each pair of state and letter, the probabilities of the
corresponding outgoing transitions sum up to 1.  PFA have been
investigated already in the 1960ies in the seminal paper of Rabin
\cite{rabin1963probabilistic}. But while the development of the theory
of automata on infinite words also started around the same time
\cite{buchi1966symposium}, the model of probabilistic automata on
infinite words has first been studied systematically in
\cite{baier2005recognizing}. The central model in this theory is the
one of probabilistic Büchi automata (PBA), which are syntactically
the same as PFA. The acceptance condition for runs is defined as for
standard nondeterministic Büchi automata (NBA): a run on an infinite
word is accepting if it visits an accepting state infinitely often
(see \cite{Thomas90,Thomas97} for an introduction to the theory of
automata on infinite words). In general, for probabilistic automata
one distinguishes different criteria of when a word is accepted. In
the positive semantics, it is required that the probability of the set
of accepting runs is greater than 0, in the almost-sure semantics it
has to be 1, and in the threshold semantics it has to be greater than
a given value $\lambda$ between 0 and 1.
It is easy to see that PFA with positive or almost-sure semantics can
only accept regular languages, because these conditions correspond to
the fact that there is an accepting run or that all runs are
accepting. For infinite words the situation is different, because
single runs on infinite words can have probability 0. Therefore, the
existence of an accepting run is not the same as the set of accepting
runs having probability greater than 0 (similarly, almost-sure
semantics is not equivalent to all runs being accepting). And in fact,
it turns out that PBA with positive (or almost-sure) semantics can
accept non-regular languages \cite{baier2005recognizing}. This
naturally raises the question under which conditions a PBA accepts a
regular language.

In \cite{baier2005recognizing} a subclass of PBA that accept only
regular languages (under positive semantics) is introduced, called
uniform PBA. The definition uses a semantic condition on the
acceptance probabilities in end components of the PBA. A syntactic
class of PBA that accepts only regular languages (under positive and
almost-sure semantics) are the hierarchical PBA (HPBA) introduced in
\cite{chadha2011randomization}. The state space of HPBA is
partitioned into a sequence of layers such that for each pair of state
and letter there is at most one transition that does not increase the
layer. Decidability and expressiveness questions for HPBA have been
studied in more detail in \cite{ChadhaSVB15,ChadhaS017}.
While HPBA accept only regular languages for positive and almost-sure
semantics, it is not very hard to come up with HPBA that accept
non-regular languages under the threshold semantics
\cite{chadha2011randomization,ChadhaSVB15} (see also the example in
Figure~\ref{fig:pwa_fig}(a) on
page~\pageref{fig:pwa_fig}). Restricting HPBA further such that there
are only two layers and all accepting states are on the first layer
leads to a class of PBA (called simple PBA, SPBA) that accept only
regular languages even under threshold semantics
\cite{chadha2011threshold}.

In this paper, we are also interested in the question under which
conditions PBA accept only regular languages. We identify syntactical
patterns in the transition structure of PBA whose absence guarantees
regularity of the accepted language. These patterns have been used
before for the classification of the degree of ambiguity of NFA and
NBA \cite{weber1991degree,rabinovich-faba,faba-dlt}. The degree of
ambiguity of a nondeterministic automaton corresponds to the maximal
number of accepting runs that a single input word can have. For NBA,
the ambiguity can (roughly) be uncountable, countable, or finite.  For
positive semantics, we show that PBA whose transition structure
corresponds to at most countably ambiguous NBA, accept only regular
languages. For almost-sure semantics, we need a slightly stronger
condition for ensuring regularity. But both classes that we identify
are easily seen to strictly subsume the class of HPBA. For the
emptiness and universality problems for these classes we obtain the
same complexities as the ones for HPBA.  In the case of threshold
semantics, we show that finite ambiguity is a sufficient condition for
regularity of the accepted language, generalizing a corresponding
result for PFA from \cite{fijalkow2017probabilistic}. The class of
finitely ambiguous PBA strictly subsumes the class of SPBA.

Besides the relation between regularity and ambiguity in PBA, we also
investigate the class of weak PBA (abbreviated PWA). In weak Büchi
automata, the set of accepting states is a union of strongly connected
components of the automaton. We show that PWA with almost-sure
semantics define the same class of languages as PBA with almost-sure
semantics (which implies that with positive semantics PWA define the
same class as probabilistic co-Büchi automata). This is in correspondence
to results for non-probabilistic automata: weak automata with
universal semantics (a word is accepted if all runs are accepting)
define the same class as Büchi automata with universal semantics, and
nondeterministic weak automata correspond to nondeterministic co-Büchi
automata (see, e.g., \cite{LodingT00}, where weak automata are called
weak parity automata). Furthermore, it is known that universal Büchi
automata, respectively nondeterministic co-Büchi automata, can be
transformed into equivalent deterministic automata (with the same
acceptance condition). An analogue of deterministic automata in the
probabilistic setting are the so-called 0/1 automata, in which each
word is either accepted with probability 0 or with probability 1. It
is known that almost-sure PBA can be transformed into equivalent 0/1
PBA (see the proof of Theorem 4.13 in
\cite{baier2012probabilistic}). Concerning weak automata, a language
can be accepted by a deterministic weak automaton (DWA) if, and only
if, it can be accepted by a deterministic Büchi and by a deterministic
co-Büchi automaton (this follows from results in
\cite{Landweber69}, see \cite{BoigelotJW05} for a more direct construction). We show an analogous
result in the probabilistic setting: The class of languages defined by
0/1 PWA corresponds to the intersection of the two classes defined by
PWA with almost-sure semantics and with positive semantics, respectively. It
turns out that this class contains only regular languages, that is,
0/1 PWA define the same class as DWA.

We also show that the regularity problem for PBA is
undecidable (the problem of deciding for a given PBA whether its
language is regular). For PBA with positive semantics this is not
surprising, as for those already the emptiness problem is undecidable
\cite{baier2012probabilistic}. However, for PBA with almost-sure
semantics the emptiness and universality problems are decidable \cite{BaierBG08,BAierBG09,chadha2011randomization}. We
show that regularity is undecidable already for PWA with almost-sure
or with positive semantics. The proof also yields that it is
undecidable for a fixed regular language whether a given PWA accepts
this language.

This work is organized as follows. After introducing basic notations in
Section~\ref{sec:prelim} we first
characterize various regular subclasses of PBA that we derive
from ambiguity patterns in Section~\ref{sec:amb}
and then we derive some related complexity results in Section~\ref{sec:compl}.
In Section~\ref{sec:weak} we present our results concerning weak probabilistic
automata and in Section~\ref{sec:conclusion} we conclude.
 \vspace{-2mm}
\section{Preliminaries}
\vspace{-2mm}
\label{sec:prelim}

First we briefly review some basic definitions. 

If $\Sigma$ is a finite alphabet, then $\Sigma^*$ is the set of all finite and
$\Sigma^\omega$ is the set of all infinite \emph{words} $w=w_0w_1\ldots$ with
$w_i\in\Sigma$. For a word $w$ we denote by $w(i)$ the $i$-th symbol $w_i$.

Classical automata used in this work have usually the shape
$(Q, \Sigma, \Delta, Q_0, F)$, where $Q$ is a finite set of states, $\Sigma$ a finite alphabet,
$\Delta \subseteq Q\times \Sigma \times Q$ is the transition relation and
$Q_0, F \subseteq Q$ are the sets of initial and final states, respectively.

We write $\Delta(p,a) := \{q \in Q \mid (p,a,q) \in \Delta \}$ to denote the set of
\emph{successors} of $p \in Q$ on symbol $a \in \Sigma$, and $\Delta(P,w)$ for
$P\subseteq Q, w\in \Sigma^*$ with the usual meaning, i.e., states reachable on word $w$
from any state in $P$.

A \emph{run} of an automaton on a word $w\in \Sigma^\omega$ is an infinite sequence of
states $q_0, q_1, \ldots$ starting in some $q_0 \in Q_0$ such that $(q_i, w(i), q_{i+1})
\in \Delta$ for all $i\geq 0$. We say that a set of runs is \emph{separated (at time $i$)}
when the prefixes of length $i$ of those runs are pairwise different.

As usual, an automaton is \emph{deterministic} if $|Q_0|=1$ and $|\Delta(p,a)|\leq 1$ for
all $p\in Q, a\in\Sigma$, and \emph{nondeterministic} otherwise. For deterministic
automata we may use a transition function $\delta:Q\times \Sigma \to Q$ instead of a
relation.

Probabilistic automata we consider have the shape $(Q,\Sigma,\delta,\mu_0,F)$, i.e., the
transition relation is replaced by a function $\delta:Q\times\Sigma\times Q\to[0,1]$ which
for each state and symbol assigns a probability distribution on successor states (i.e.
$\sum_{q\in Q}\delta(p,a,q) = 1$ for all $p\in Q, a\in\Sigma$), and $\mu_0:Q\to [0,1]$ with
$\sum_{q\in Q}\mu_0(q) = 1$ is the initial probability distribution on states.
The \emph{support} of a distribution $\mu$ is the set $\op{supp}(\mu):=\{ x \mid \mu(x)>0 \}$.
Similarly as above, we may write $\delta(\mu,w)$ and mean the resulting probability distribution
after reading $w\in \Sigma^*$, when starting with probability distribution $\mu$.

For a probabilistic automaton $\mc{A}$ the \emph{underlying automaton}
$\underlyingnba{\mc{A}}$ is given by recovering the transition relation
$\Delta := \{(p,x,q) \mid \delta(p,x,q)>0 \}$ of positively reachable states and the
initial state set $Q_0 := \op{supp}(\mu_0)$.

As usual, a run of an automaton for finite words is \emph{accepting} if it ends in a final
state. For automata on infinite words, run acceptance is determined by the Büchi (run visits
infinitely many final states) or Co-Büchi (run visits finitely many final states) conditions.

We write $p \overset{x}{\to} q$ if there exists a path from $p$ to $q$ labelled
by $x\in \Sigma^{+}$ and $p \to q$ if there exists some $x$ such that $p\overset{x}{\to} q$.
The \emph{strongly connected component (SCC)} of $p\in Q$ is $\scc(p) := \{q\in Q \mid p=q
\text{\ or\ } p \to q \text{\ and\ } q \to p \}$.
The set $\op{SCCs}(\mc{A}) := \{\scc(q) \mid q\in Q\}$ is the set of all SCCs and partitions $Q$.
An SCC is \emph{accepting} (\emph{rejecting}) if all (no) runs that stay there forever are
accepting. An SCC is \emph{useless} if no accepting run can continue from there.
An automaton is \emph{weak}, if the set of final states is a union of its SCCs. In this
case, Büchi and Co-Büchi acceptance are equivalent and we treat weak automata as Büchi
automata.

A classical automaton is \emph{trim} if it has no useless SCCs, whereas a probabilistic
automaton is trim if it has at most one useless SCC, which is a rejecting sink that we
canonically call $q_{rej}$.
We assume w.l.o.g. that all considered automata are trim, which also means that in an
underlying automaton the sink $q_{rej}$ is removed.

We call transitions of probabilistic automata that have probability 1 \emph{deterministic} and
otherwise \emph{branching}.
If there are transitions $p \overset{a}{\to} q$ and $p \overset{a}{\to} q'$ with $q\neq
q'$, we call this pattern a \emph{fork}. Every branching transition clearly has at least
one fork. We call a $(p,q,q')$ fork \emph{intra-SCC}, if $p,q,q'$ are all in the same SCC,
otherwise it is an \emph{inter-SCC} fork. A run of an automaton is \emph{deterministic}
if it never goes through forks, and \emph{limit-deterministic} if it goes only through
finitely many forks. We say that two deterministic runs \emph{merge} when they reach the
same state simultaneously.
For a finite run prefix $\rho$, we
call all valid runs with this prefix \emph{continuations} of $\rho$.

A classical automaton $\mc{A}$ \emph{accepts} $w\in \Sigma^\omega$ if there
exists an accepting run on $w$, and the language $L(\mc{A})$
\emph{recognized} by $\mc{A}$ is the set of all
accepted words. If $P$ is a set of states of an automaton, we write
$L(P)$ for the language accepted by this automaton with initial state
set $P$. For sets consisting of one state $q$, we write $L(q)$ instead
of $L(\{q\})$.

For a probabilistic automaton $\mc{A}$ and an input word $w$ (finite
or infinite), the transition structure of $\mc{A}$ induces a
probability space on the set of runs of $\mc{A}$ on $w$ in the usual
way. We do not provide the details here but rather refer the reader
not familiar with these concepts to \cite{baier2012probabilistic}. In
general, we write $\op{Pr}(E)$ for the probability of a measurable
event $E$ in a probability space.
For probabilistic automata, we consider \emph{positive}, \emph{almost-sure} and
\emph{threshold} semantics, i.e., an automaton accepts $w$ if the probability of
the set of accepting runs on $w$ is $>0$, ${=}1$ or ${>}\lambda$ (for
some fixed $\lambda\in ]0,1[$),
respectively. For an automaton $\mc{A}$ these languages are denoted by $L^{>0}(\mc{A}),
L^{=1}(\mc{A})$ and $L^{>\lambda}(\mc{A})$, respectively, whereas $L(\mc{A}) := L(\underlyingnba{\mc{A}})$
is the language of the underlying automaton. A probabilistic automaton is 0/1
if all words are accepted with either probability 0 or 1 (in this case the languages with
the different probabilistic semantics coincide).

To denote the type of an automaton, we use abbreviations of the form
XYA$^{(\gamma)}$ where the type of transition structure is denoted by
X $\in\{$ D (det.), N (nondet.), P (prob.)  $\}$, the acceptance
condition is specified by Y $\in\{$ F (finite word), B (Büchi), C
(Co-Büchi), W (Weak) $\}$, and for probabilistic transitions the
semantics for acceptance is given by $\gamma\in\{$>0,\!\!  =1,\!\!
>$\lambda, 0/1 \}$.

By $\mbb{L}^{(\gamma)}(\mathsf{XYA})$ we denote the whole class of
languages accepted by the corresponding type of automaton. If
$\mbb{L}$ is a set of languages, then $\overline{\mbb{L}}$ denotes the
set of all complement languages (similarly, for a language $L$, we
denote by $\overline{L}$ its complement), and $\op{BCl}(\mbb{L})$ the
set of all finite boolean combinations of languages in $\mbb{L}$.
We use the notion of \emph{regular language} for finite words and for infinite words (the type of words is always clear from the context).

 \section{Ambiguity of PBA}
\label{sec:amb}

Ambiguity of automata refers to the number of different accepting runs on a word or on all words.
An automaton is \emph{finitely ambiguous} (on $w$) if there are at most $k$ different
accepting runs (on $w$) for some fixed $k \in \Nat$, and in case of at most one accepting run
it is called \emph{unambiguous}. If on each word there are only finitely
many accepting runs, but no constant upper bound over all words, then it is
\emph{polynomially ambiguous} if the number of different run prefixes that are possible for any
word prefix of length $n$ can be bounded by a polynomial in $n$, and otherwise
\emph{exponentially ambiguous}. Finally, if if there exist words that have infinitely many
runs, but no word on which there are uncountably many accepting runs, then it is
\emph{countably ambiguous}, and otherwise it is \emph{uncountably ambiguous}.

In \cite{faba-dlt} (see also \cite{rabinovich-faba}), a syntactic characterization of those classes is presented for NBA
by simple patterns of states and transitions. We define those patterns
here and refer to \cite{faba-dlt} for further details. An automaton $\mc{A}$ has an
\emph{IDA pattern} if there exist two states $p\neq q$ and a word $v\in\Sigma^*$ such that
$p\overset{v}{\to}p$, $p\overset{v}{\to}q$ and $q\overset{v}{\to}q$. If additionally $q\in F$, then
this is also an $\IDAF$ pattern. Finally, $\mc{A}$ has an \emph{EDA pattern} if there exists
a state $p$ and $v \in \Sigma^*$ such that there are two different paths
$p\overset{v}{\to} p$, and if additionally $p\in F$, this is also an $\EDAF$ pattern.
If a PBA has no $\EDA$ pattern, we call it \emph{flat}, reflecting the naming of a similar
concept in other kinds of transition systems (e.g. \cite{leroux2004flatness}).
The names IDA and EDA abbreviate ``infinite/exponential degree of ambiguity'', which they
indicated in the original NFA setting, and we keep those names for consistency.

By $k$-\NBA, $n^k$-\NBA, $2^n$-\NBA, $\aleph_0$-\NBA we denote the subsets of at most finitely,
polynomially, exponentially and countably ambiguous NBA (and similarly for other types of automata).
When speaking about ambiguity of some PBA $\mc{A}$, we mean the ambiguity of the trimmed underlying NBA
$\underlyingnba{\mc{A}}$.

In \cite{chadha2011randomization}, hierarchical PBA (HPBA) were identified as a
syntactic restriction on PBA which ensures regularity under positive and almost-sure semantics. A PBA
with a unique initial state is hierarchical, if it admits a ranking on the states such that at
most one successor on a symbol has the same rank, and no successor has a smaller rank. A
HPBA has $k$ levels if it can be ranked with only $k$ different values. Simple PBA (SPBA)
were introduced in \cite{chadha2011threshold} and are restricted HPBA with two levels
such that all accepting states are on level 0.

\begin{figure}
  \centering
  
\begin{tikzpicture}
\pgftransformscale{.65}

\draw[very thick] (-4,0) parabola bend (1.5,5.8) (8,0);
\draw[very thick] (-8,0) parabola bend (-1.5,5.8) (4,0);
\fill[color=white] (-4,0) parabola bend (1.5,5.75) (8,0);
\fill[color=white] (-8,0) parabola bend (-1.5,5.75) (4,0);

\draw[very thick] (9.2,0) -- (-9.2,0);

\draw[dotted] (-1.1,0) parabola bend (0,0.7) (1.1,0) ;
\node at (0,0.3) {
    \begin{tabular}{c} $\SPBA$ \end{tabular}
  };

\draw[dashed] (-2,0) parabola bend (0,1.5) (2,0);
\node at (0,1.1) {
    \begin{tabular}{c} unamb. \end{tabular}
  };

\draw[very thick] (-2.5,0) parabola bend (0,3) (2.5,0);
  \node at (0,2.2) {
    \begin{tabular}{c}
      $\lnot \IDA$ \\
      fin. amb.
    \end{tabular}
  };

\node at (0,3.6) {
    \begin{tabular}{c}
      $\lnot \EDA, \lnot \IDAF$ \\
      poly. amb.
    \end{tabular}
  };

\draw (-4,0) parabola bend (1.5,5.8) (8,0);
  \node at (3.8,3.6) {
    \begin{tabular}{c}
      $\lnot \IDAF$ \\
      exp. amb.
    \end{tabular}
  };

\draw (-8,0) parabola bend (-1.5,5.8) (4,0);
  \node at (-3.8,3.6) {
    \begin{tabular}{c}
      $\lnot \EDA$ \\
      flat
    \end{tabular}
  };

\draw[very thick] (-9,0) parabola bend (0,7.2) (9,0);
  \node at (0,6.5) {
\begin{tabular}{c}
      $\lnot \EDAF$ \\
      countably amb.
    \end{tabular}
  };

\draw[dotted] (-8,0) parabola bend (-5,2) (3,0);
  \node at (-5,1.5) {
    \begin{tabular}{c}
      $\HPBA$ \\
\end{tabular}
  };

\node[align=center] at (-5.8,6.8) { $\mbb{L}^{>0}(\aleph_{0}\text{-}\PBA)$ \\ regular };
\node[align=center]  at (6.3,6.8) { $\mbb{L}^{=1}(flat\ \PBA \cup 2^k\text{-}\PBA)$ \\ regular };
\node[align=center]  at (-8,4) { $\mbb{L}^{>\lambda}(k\text{-}\PBA)$ \\ regular };

\draw (-5.8,6.1) edge[->,decorate,decoration={snake,amplitude=.4mm,segment length=2mm,post length=1mm}] (-4.6,5.4);
\draw (6.3,6.1) edge[->,decorate,decoration={snake,amplitude=.4mm,segment length=2mm,post length=1mm}] (4.6,4.6);
\draw (-8,3.3) edge[->,decorate,decoration={snake,amplitude=.4mm,segment length=2mm,post length=1mm}] (-1.5,2);
\end{tikzpicture}
 \caption{Illustration of the automata classes with restricted ambiguity as presented for
NBA in \cite{faba-dlt}, which are characterized by
the absence of the state patterns $\IDA, \IDAF, \EDA,$ and $\EDAF$ and their relation to the
restricted classes called ``Hierarchical PBA'' ($\HPBA$) \cite{chadha2011randomization} and ``Simple
PBA'' ($\SPBA$) \cite{chadha2011threshold}. We identify classes in this hierarchy which
can be seen as extensions ``in spirit'' of respectively
SPBA and HPBA, subsuming them while also preserving their good properties, as
e.g. definition by syntactic means, regularity under different semantics and several
complexity results.
}
\label{fig:hpba_and_amb}
\end{figure}
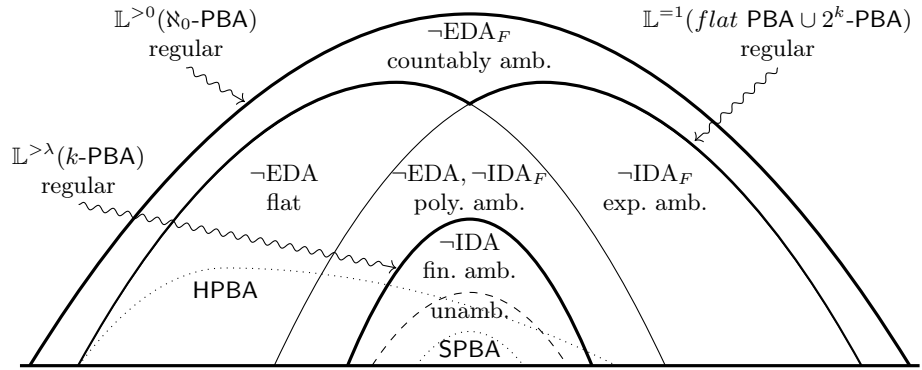

First, we show how HPBA relate to the ambiguity hierarchy,
which can easily be derived by inspection of the definitions.
A visual illustration is given in Figure~\ref{fig:hpba_and_amb}.

\begin{restatable}[Relation of HPBA and the ambiguity hierarchy]{proposition}{prphpbaambiguity}
  \label{prp:hpba_ambiguity}
  \ \\[-4mm]
  \begin{enumerate}
    \item $\HPBA \subset $ flat PBA $\subset \aleph_0\text{-}\PBA$.
    \item $k\text{-}\PBA \not\subseteq \HPBA$ and $\HPBA \not\subseteq k\text{-}\PBA$.
    \item $\SPBA \subset$ unambiguous PBA $\subset k\text{-}\PBA$.
  \end{enumerate}
\end{restatable}

Starting from these observations, this work was motivated by the question whether the
ambiguity restrictions, which were only implicit in HPBA and SPBA, can be used explicitly
to get larger classes with good properties. In the following we will positively answer this
question.

\subsection{From classical to probabilistic automata}

First, we observe that probabilistic automata can recognize regular languages
even under severe ambiguity restrictions.

\begin{proposition}
  \label{prp:dba_to_kpba}
  Let $\mc{A}$ be a DBA. Then there exists an unambiguous PBA $\mc{B}$ such that
  $L^{>0}(\mc{B}) = L^{=1}(\mc{B}) = L(\mc{A})$.
\end{proposition}
\begin{proof}
  As $\mc{A}$ is a (w.l.o.g. complete) DBA, there exists exactly one run on each word and all
  transitions when seen as PBA must have probability 1. Clearly this unique natural 0/1 PBA obtained from $\mc{A}$
  accepts the same language under both probable and almost-sure semantics and it is
  trivially unambiguous.
  \lncsqed
\end{proof}

Limit-deterministic NBA (LDBA) are NBA which are deterministic in all non-rejecting SCCs.
The natural mapping of LDBA into PBA \cite[Lemma 4.2]{baier2012probabilistic} already
trivially yields countably ambiguous automata (because the deterministic part of the LDBA
cannot contain an EDA$_F$ pattern, which implies uncountable ambiguity \cite{faba-dlt}).
The following result shows that already unambiguous
PBA under positive semantics suffice for all regular languages.

\begin{restatable}{theorem}{thmnbatokpba}
  \label{thm:nba_to_kpba}
Let $L\subseteq \Sigma^\omega$ be a regular language.\\
  Then there exists an unambiguous PBA $\mc{B}$ such that $L^{>0}(\mc{B}) = L$.
\end{restatable}
\begin{proof}[sketch] Let $\mc{A}=(Q,\Sigma,\delta,q_0,c)$ be a deterministic parity automaton accepting $L$,
  i.e., a finite automaton with priority function $c:Q\to\{1,\ldots,m\}$ such that $w \in
  L(\mc{A})$ iff the smallest priority assigned to a state on the unique run of $\mc{A}$
  on $w$ which is seen infinitely often is even.

  We construct an unambiguous LDBA for $L$, which then easily yields a PBA$^{>0}$
  by assigning arbitrary probabilities (\cite[Lemma 4.2]{baier2012probabilistic}) without
  influencing the ambiguity.
If the parity automaton $\mc{A}$ has $m$ priorities, the LDBA $\mc{B}$ can be obtained
  by taking $m+1$ copies, where $m$ of them are responsible for one priority each, and one is
  modified to guess which priority $i$ on the input word is the most important one appearing infinitely often along the run of $\mc{A}$, and correspondingly switch
  into the correct copy. This switching is done unambiguously for the first position after which no priority more important than $i$ appears.
  \lncsqed
\end{proof}

\subsection{From probabilistic to classical automata}

First we establish a result for flat PBA, i.e. PBA that have no $\EDA$ pattern. In automata without
$\EDA$ pattern there are no states which are part of two different cycles labeled by the
same finite word.
Even though we defined flat PBA by using an ambiguity pattern, the
set of flat PBA does not correspond to an ambiguity class, but it is
useful for our purposes due to the following property:

\begin{lemma}
  \label{lem:flat_pba_limdet}
  If $\mc{A}$ is a flat PBA and $w\in \Sigma^\omega$, then the probability of
  a run of $\mc{A}$ on $w$ to be limit-deterministic is 1.

\end{lemma}
\begin{proof}

  Let $\op{Runs}(\mc{A},w)$ denote the set of all runs of $\mc{A}$ on $w$ and $\op{nldRuns}(\mc{A},w)$
  denote the subset containing all such runs that are not limit-deterministic.
  As $\mc{A}$ is flat, it has no $\EDA$ and thus also no $\EDAF$ pattern, hence
  $\mc{A}$ is at most countably ambiguous (by \cite{faba-dlt}).
  Moreover, there are not only at most countably many accepting runs on any word,
  but also countably many rejecting runs (which can be seen by a simple
  generalization of \cite[Lemma~4]{faba-dlt}). But as all runs are disjoint events, each
  run $\rho$ that uses infinitely many forks has probability 0, and the
  total number of runs is countable, we can see that
\[ \op{Pr}(\op{Runs}(\mc{A},w) \setminus \op{nldRuns}(\mc{A},w)) =
  \smashoperator{\sum_{\rho \in \op{Runs}(\mc{A},w)}} \op{Pr}(\rho)\quad -\quad
  \smashoperator{\sum_{\rho \in \op{nldRuns}(\mc{A},w)}} \op{Pr}(\rho) = 1 - 0 = 1.
  \quad\quad
  \lncsqed
  \]
\end{proof}

The following lemma characterizes acceptance of PBA under extremal semantics with
restricted ambiguity and is crucial for the constructions in the following sections:
\begin{lemma}[Characterizations for extremal semantics]
  \label{lem:cntamb_pba_runs}\quad \\
  Let $\mc{A}$ be a PBA.
  \begin{enumerate}
    \item If $\mc{A}$ is at most countably ambiguous, then \\
      $w\in L^{>0}(\mc{A}) \Leftrightarrow$ there exists an accepting run on $w$ that is limit-deterministic.
    \item If there are finitely many accepting runs of $\mc{A}$ on $w$, then \\ $w\in L^{=1}(\mc{A}) \Leftrightarrow$ all runs on $w$ are accepting and limit-deterministic.
\item  If $\mc{A}$ is flat, then \\
      $w\in L^{=1}(\mc{A}) \Leftrightarrow$ there is no limit-deterministic rejecting run on $w$.
  \end{enumerate}
\end{lemma}

\begin{proof}
  $(1.):$
  For contradiction, assume that every accepting run on $w$ goes through forks infinitely often.
  But then the probability of every individual accepting run on $w$ is 0. Each run is a
  measurable event (it is a countable intersection of finite prefixes) and clearly
  disjoint from other runs, as two different runs must eventually differ after a finite
  prefix. But as the number of accepting runs is countable by assumption, by
  $\sigma$-additivity it follows that the probability of all accepting runs is also 0,
  contradicting the fact that $w\in L^{>0}(\mc{A})$.

  For the other direction, pick a limit-deterministic accepting run $\rho$ of $\mc{A}$ on
  $w$ and let $uv=w$ and $q\in Q$ such that the state of $\rho$ after reading $u$ is $q$ and
  there are no forks visited on $v$. Clearly, the probability to be in $q$
  after $u$ in a run of $\mc{A}$ is positive (because $u$ is finite), and the probability
  that $\mc{A}$ continues like $\rho$ from $q$ on $v$ is $1$. Hence, the probability of
  $\rho$ is positive.

  $(2.):$
The $(\Leftarrow)$ direction is obvious. We now proceed to show $(\Rightarrow)$.
  Take some time $t$ after which all accepting runs on $w$ separated.
  Assume that some accepting run $\rho$ is not limit-deterministic.
  But then $\rho$ goes through infinitely many forks after $t$ which
  with positive probability lead to a successor from which the probability to accept is 0,
  and the probability of following $\rho$ is also 0. As the probability to follow $\rho$
  until time $t$ is positive, but after that the probability to accept is 0, this implies
  that there is a positive probability that $\mc{A}$ rejects $w$. Therefore, all accepting
  runs on $w$ must be limit-deterministic.
  Now assume that some run $\rho$ on $w$ is rejecting. Following this run until the time at which
  $\rho$ is separated from all accepting runs has positive probability and all
  continuations must be also rejecting, so $\mc{A}$ must reject $w$.

  $(3.):$
  Clearly $(\Rightarrow)$ holds, because a limit-deterministic rejecting run has positive
  probability, i.e., if such a run exists on $w$, then $\mc{A}$ cannot accept almost surely. For
  $(\Leftarrow)$, observe that
because $\mc{A}$ is flat, we know by Lemma~\ref{lem:flat_pba_limdet} that with probability 1
  runs are limit-deterministic. Hence, if there exists no limit-deterministic rejecting
  run on $w$ (which would have positive probability), then with probability 1 runs are
  limit-deterministic and accepting.
  \lncsqed

\end{proof}

Using these characterizations, we can provide simple constructions from probabilistic
to classical automata.

\begin{restatable}{theorem}{thmpscpbatonba}
  \label{thm:ps_cpba_to_nba}
  Let $\mc{A}$ be a PBA that is at most countably ambiguous.

  Then $L^{>0}(\mc{A})$ is a regular language.
\end{restatable}
\begin{proof}[sketch]
An NBA construction taking two copies of the PBA, where in the first copy no
state is accepting and the second copy has no forks, with the purpose of guessing
a limit-deterministic accepting run.
\lncsqed
\end{proof}

\begin{corollary}
  If $L^{>0}(\mc{A})$ is not regular, then it contains an $\EDAF$ pattern.
\end{corollary}

\begin{restatable}{theorem}{thmasefpbatodba}
  \label{thm:as_efpba_to_dba}
  Let $\mc{A}$ be a PBA that is at most exponentially ambiguous or flat.

  Then $L^{=1}(\mc{A})$ is regular and recognizable by DBA.
\end{restatable}
\begin{proof}[sketch]
  Both cases (exp. ambiguous or flat) shown using a deterministic
  breakpoint construction resulting in a DBA.
  In one case it checks whether all runs are accepting, in
  the other it checks that there are no limit-deterministic rejecting runs.
  \lncsqed
\end{proof}

\begin{corollary}
  If $L^{=1}(\mc{A})$ is not regular, \\
  then $\mc{A}$ contains both an $\EDA$ and an $\IDAF$ pattern.
\end{corollary}

The corollaries above follow directly from the theorems and the
syntactic characterization of ambiguity classes \cite{faba-dlt}. The
following proposition states that these characterizations of
regularity in terms of the ambiguity patterns are tight.

\begin{figure}[htbp]
  \centering
  \begin{center}
  (a)\
  \begin{tikzpicture}[baseline={([yshift=-.5ex]current bounding box.center)},shorten >=1pt,
    node distance=1.2cm,inner sep=1pt,on grid,auto]
    \node[state,initial, initial text={$\frac{1}{2}$}] (qa)   {$q_a$};
    \node[state,initial, initial text={$\frac{1}{2}$},below=of qa] (qb)   {$q_b$};
    \node[state,right=of qa] (qp) {$q_+$};
    \node[state,accepting,right=of qb] (qs) {$q_\$$};
      \path[->]
      (qa) edge [loop above] node[xshift=10mm,yshift=-3mm] {$b:1,a:\frac{1}{2}$} (qa)
      (qa) edge [swap] node[yshift=-2pt] {$a:\frac{1}{2}$} (qp)
      (qp) edge [loop right] node {$a,b$} (qp)
      (qb) edge [loop below] node[xshift=10mm,yshift=3mm] {$a:1,b:\frac{1}{2}$} (qb)
      (qb) edge [] node {\$} (qs)
      (qp) edge [] node {\$} (qs)
      (qs) edge [loop right] node {\$} (qs)
      ;
  \end{tikzpicture}
  \quad(b)\
    \begin{tikzpicture}[baseline={([yshift=-.5ex]current bounding box.center)},shorten >=1pt,
      inner sep=1pt,node distance=1.2cm,on grid,auto]
      \node[state,initial below,accepting,initial text={}] (q0)   {$q_0$};
      \node[state] (q1) [right=of q0] {$q_1$};
      \path[->]
      (q0) edge [loop above] node {$a:1\text{-}\lambda$} (q1)
      (q0) edge [bend left] node {$a:\lambda$} (q1)
      (q1) edge [bend left] node {$b$} (q0)
      (q1) edge [loop above] node {$a$} (q1);
    \end{tikzpicture}
  \quad(c)\
  \begin{tikzpicture}[baseline={([yshift=-.5ex]current bounding box.center)},shorten >=1pt,
    node distance=1.2cm,inner sep=1pt,on grid,auto]
    \node[state,initial left, initial text={}] (q0)   {$q_0$};
    \node[state] (q1) [right=of q0] {$q_1$};
    \node[state] (q2) [below=of q0] {$q_2$};
    \node[state,accepting] (qf) [right=of q2] {$q_f$};
      \path[->]
      (q0) edge[bend left] node {$a:\lambda$} (q1)
      (q1) edge[bend left,swap] node {$b$} (q0)
      (q0) edge[swap] node {$a:(1\text{-}\lambda)$} (q2)
      (q1) edge [loop right] node {$a$} (q1)
      (q2) edge [loop left] node {$a:(1\text{-}\lambda)$} (q2)
      (q2) edge [swap] node [near end] {$a:\lambda$} (q1)
      (q2) edge [swap] node {$b$} (qf)
      (qf) edge [loop right] node {$\Sigma$} (qf)
;
  \end{tikzpicture}
\end{center}
   \caption{
    (a) Some PWA which accepts the non-regular language
      $\{\ w = (a+b)^*\$^\omega \mid \#_a(w) > \#_b(w)\ \}$ with a threshold of
      $\frac{1}{2}$, where $\#_x(w)$ denotes the number of occurrences of $x\in \Sigma$ in $w\in \Sigma^\omega$.
    (b) A family of PBA ${\mc{P}}_\lambda$ from \cite{baier2012probabilistic} such that
      $\mathbb{L}^{>0}({\mc{P}}_\lambda)$ is not regular for any $\lambda \in \mathbb{R}$.
    (c) A family of PWA $\tilde{\mc{P}}_\lambda$  (closely related to \cite[Fig. 6]{baier2012probabilistic})
    such that $\mathbb{L}^{=1}(\tilde{\mc{P}}_\lambda)$ is not regular for any $\lambda \in \mathbb{R}$.
  }
  \label{fig:pwa_fig}
\end{figure}
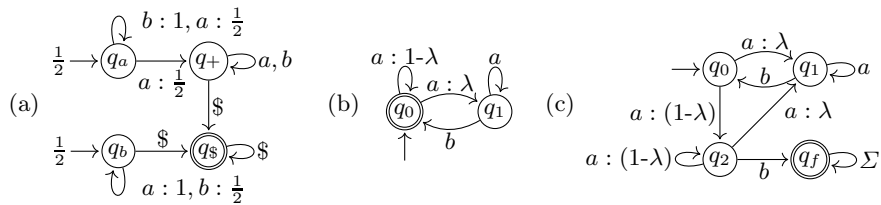

\vspace{1.3cm} \begin{proposition} \quad \label{prp:amb_counterexamples}
  There exist PBA...\\[-6mm]
  \begin{enumerate}
    \item ...with $\EDAF$ pattern (i.e. uncountably ambiguous) that accept \\
      non-regular languages under positive semantics.
    \item ...with no $\EDAF$ pattern (i.e. countably ambiguous) that accept \\
      non-regular languages under almost-sure semantics.
  \end{enumerate}
\end{proposition}
\begin{proof}
(1.) Note that this statement just means that there are PBA accepting non-regular languages, which is well known.
  For example, the automata family from \cite[Fig. 3]{baier2012probabilistic}, depicted in
  Figure~\ref{fig:pwa_fig}(b), accepts non-regular languages under positive
  semantics and clearly contains an $\EDAF$ pattern, e.g. there are two different paths
  from $p_0$ to $p_0$ on the word $aab$.

(2.)
  The automata family depicted in Figure~\ref{fig:pwa_fig}(c) is a simple
  modification of the PBA family depicted in \cite[Fig. 6]{baier2012probabilistic} and
  recognizes the same non-regular languages under almost-sure semantics.
It does not contain an $\EDAF$ pattern, because the accepting state is a sink, but
  it does contain an $\IDAF$ and an $\EDA$ pattern (both e.g. on $aab$), so it is countably ambiguous and not flat.
\lncsqed
\end{proof}

This completes our classification of regular subclasses of PBA
under extremal semantics that are defined by ambiguity patterns, showing that going beyond
the restricted classes presented above (by allowing more patterns) in general leads to a
loss of regularity.

Notice that the presented constructions do not track exact probabilities, just whether
transitions have a probability $>0$ or $=1$.
This is a noteworthy observation, as in general, the probabilities do
matter for PBA, as shown in \cite[Thm. 4.7, Thm. 4.11]{baier2012probabilistic}.
\begin{proposition}
  Let $\mc{A}$ be a PBA.
  The exact probabilities in $\mc{A}$ do not influence $L^{>0}(\mc{A})$ if $\mc{A}$ is at most
  countably ambiguous, and $L^{=1}(\mc{A})$ if $\mc{A}$ is at most exponentially ambiguous or
  flat.
\end{proposition}

\subsection{Threshold Semantics}

In this section we consider PBA under threshold semantics and we will see that in this
setting, we lose regularity much earlier than in the case of extremal semantics,
but there is still the large and natural subclass of finitely ambiguous PBA that retains
regularity. Before we can show this, we need to derive a suitable
characterization of such languages.

We derive it from the following simple observation, which was also
used more implicitly in the proof that Simple HPBA with threshold semantics are equivalent
to DBA in \cite{chadha2011threshold}.
\begin{lemma}
  \label{lem:run_finvals_above_th}
  Let $\mc{A}$ be a PBA. Then for every threshold $\lambda \in ]0,1]$,
    there exists a finite set of probability values $V_{\geq\lambda} \subset [\lambda,1]$ such that
      for every finite run prefix with probability $v$ in $\mc{A}$ we have $v \geq \lambda
      \Rightarrow$ $v \in V_{\geq\lambda}$.
\end{lemma}
\begin{proof}
  Observe that given a finite set of real numbers $R \subset [0,1]$, the set
  $R_{\geq\lambda} := \{ r \mid r = \prod_i r_i \geq \lambda,\ r_i\in R \}$ must be finite, as
  in any sequence $p_1p_2\ldots$ of $p_i\in R$, only at most
  $m = \lceil \log_\lambda(\max R) \rceil$ values can be $<1$
  and such that the product of the sequence remains $\geq \lambda$.
  In our case, let $R$ be the set of distinct probabilities assigned to edges (including
  the initial edges) in $\mc{A}$. As every finite run prefix by definition has the
  probability given by the product of the edge probabilities, this implies the statement.
  \lncsqed
\end{proof}

If there is just one accepting run (i.e., the automaton is unambiguous), one can easily
construct a nondeterministic automaton that guesses an accepting run and tracks it along
with its probability value, of which there are only finitely many above the threshold. In
the case that there are multiple accepting runs, for acceptance only the sum of their
probabilities matters. As individual runs can in principle have arbitrarily small
probability values, it is not obvious that the same approach (tracking a set of runs) can
work. Determining a suitable cut-off point is not as simple, because it is not apparent
when a single run becomes so improbable that it does not matter among the others. However,
we will now show that such a cut-off point must exist:

\begin{lemma}
  \label{lem:cutoff_fpba_th}
  Let $\mc{A}$ be a PBA, $\lambda \in ]0,1]$ a threshold and $k\in \mathbb{N}$.
  There exists $\varepsilon_k \in\ ]0,\lambda]$ such that
  for all sets $R^t=\{\rho^t_i\}_{i=1}^j$ of at most $j\leq k$ different run prefixes in $\mc{A}$ of the
  same length $t\in \mathbb{N}$, $\op{Pr}(R^t) = \sum_{i=1}^j \op{Pr}(\rho^t_i) < \lambda$
  implies that $\op{Pr}(R^t) < \lambda - \varepsilon_k$.
\end{lemma}
\begin{proof}
  We prove this by induction on the number of runs $k$. For $k=1$, i.e. a single run
  prefix, let $V_{\geq \lambda}$ be the finite (by Lemma~\ref{lem:run_finvals_above_th})
  set of different probability values $\geq \lambda$ and let $E$ be the set of distinct
  probabilities in the automaton $\mc{A}$. Then clearly
  $v_{\max,<\lambda}:= \max \{ a\cdot b \mid a\cdot b  < \lambda, a\in V_{\geq \lambda}, b\in E \}$
  is the largest probability value $< \lambda$ that can correspond to a finite run prefix in $\mc{A}$.
  Hence, we can just pick an $\varepsilon_1 < \lambda - v_{\max,<\lambda}$ and immediately
  get that for any run prefix with probability $v < \lambda$, we have that
  $v \leq v_{\max,<\lambda} < \lambda - \varepsilon_1$.

  Now assume the statement holds for all sets with at most $k$ run prefixes. Let $R^t$
  be a set of $k+1$ of different run prefixes of the same length such that $\op{Pr}(R^t) <
  \lambda$ and let $\varepsilon := \varepsilon_k$. Then we know
  that for every subset $S$ of at most $k$ runs of $R^t$ we have $\op{Pr}(S) < \lambda -
  \varepsilon$. Also, every single run prefix can by Lemma~\ref{lem:run_finvals_above_th}
  have one of only finitely many probability values in $V_{\geq \varepsilon}$ that are $\geq
  \varepsilon$ and there exists a value $v_{\max,<\varepsilon}$ denoting the largest
  possible probability value $<\varepsilon$ that a single run prefix can have.

  If there exists a run prefix $\rho \in R^t$ with probability value $v < \varepsilon$,
  then we know that
  $\op{Pr}(R^t) = \op{Pr}(R^t \setminus\{\rho\}) + v
  < (\lambda - \varepsilon) + v_{\max,<\varepsilon} < \lambda$.
If every run in $R^t$ has a probability value $\geq \varepsilon$,
  then every run prefix in $R^t$ has as probability one of the values in $V_{\geq
  \varepsilon}$. Consider all sums of $k$ values from $V_{\geq\varepsilon}$, which are
  finitely many, and pick the largest sum $s$ which is $<\lambda$.
  Choose $\varepsilon_{k+1}$ such that $\varepsilon_{k+1} < \min(\varepsilon
  - v_{\max,<\varepsilon}, \lambda - s)$ to account for both cases.
  \lncsqed
\end{proof}

From this we can derive the following characterization of languages accepted by finitely
ambiguous PBA under threshold semantics:
\begin{lemma}
  \label{lem:th_kpba_characterization}
  Let $\mc{A}$ be a $k$-ambiguous PBA and $\lambda \in ]0,1]$ a threshold.
  There exists  an $\varepsilon \in\ ]0,\lambda]$  such that for all $w \in \Sigma^\omega$:
  $w\in L^{>\lambda}(\mc{A})$ iff
  there exists
  a set $R$ of limit-deterministic accepting runs of $\mc{A}$ on $w$
  with $\op{Pr}(R) > \lambda$, $\op{Pr}(S) \leq \lambda$ for all $S\subset R$ and at most
  one run $\rho\in R$ with $\op{Pr}(\rho) < \varepsilon$.
\end{lemma}
\begin{proof}
  Clearly $(\Leftarrow)$ holds, as then $w$ is accepted with probability
  $\geq \op{Pr}(R) > \lambda$. We now show $(\Rightarrow)$.
In a finitely ambiguous PBA there are only finitely many different accepting runs on each word.
Furthermore, as after finite time all accepting runs have separated and each accepting
  run that visits forks infinitely often has probability 0, accepting runs that visit
  forks infinitely often do not contribute positively to the acceptance probability and
  thus can be ignored.
  Hence, if $w\in L^{>\lambda}(\mc{A})$, there is a number of accepting runs that
  eventually all become deterministic and each such run has a positive probability, which
  must in total be $>\lambda$.

  Let $R$ be a
  set of different limit-deterministic accepting runs of $\mc{A}$ on
  $w$ such that $\op{Pr}(R) > \lambda$ and $\op{Pr}(S) \leq \lambda$ for all
  $S\subset R$. As there are only finitely many accepting runs, such a set $R$ must exist.
  Furthermore, notice that each limit-deterministic run has a finite prefix which has the
  same probability as the whole run, so there exists a time $t$ such that the probability
  of the set of all different prefixes of runs in $R$ of length $t$ is exactly
  $\op{Pr}(R)$, so that Lemma~\ref{lem:cutoff_fpba_th} applies.

  Now pick an $\varepsilon := \varepsilon_k$ given by Lemma~\ref{lem:cutoff_fpba_th}.
  We claim that at most one run $\rho \in R$ can have a probability less than $\varepsilon$.
  If there is no such run in $R$, we are done. Otherwise let $\rho$ be a run with
  $\op{Pr}(\rho) =: p < \varepsilon$ and notice that by choice of $R$, we have that
  $\op{Pr}(R\setminus\{\rho\}) =: s \leq \lambda$. It cannot be the case that
  $s < \lambda$, as then by Lemma~\ref{lem:cutoff_fpba_th} we have $s < \lambda - \varepsilon$,
  which implies that $\op{Pr}(R)= s+p < \lambda$, which is a contradiction.
  Hence, now assume that $s = \lambda$. But then, if there is any $\rho' \neq \rho \in R$
  such that $\op{Pr}(\rho') =: p' < \varepsilon$, by the same argument we get the
  contradiction that $s - p' < \lambda - \varepsilon$ and hence $s < \lambda$. Therefore,
  no other run in $R$ can have a probability $<\varepsilon$.
  \lncsqed
\end{proof}

Now we can perform the intended automaton construction to show:
\begin{restatable}{theorem}{thmthkpbatonba}
\label{thm:th_kpba_to_nba}$L^{>\lambda}(\mc{A})$ is regular for each $k$-ambiguous PBA $\mc{A}$ and $\lambda \in ]0,1[$.
\end{restatable}
\begin{proof}[sketch] We use the characterization of Lemma~\ref{lem:th_kpba_characterization} to
  construct a generalized Büchi automaton accepting $L^{>\lambda}(\mc{A})$.
  Intuitively, the new automaton just guesses at most $k$ different runs of
  $\mc{A}$ and verifies that the guessed runs are limit-deterministic and accepting.
  The automaton additionally tracks the probability of the runs over time, to determine
  whether the individual runs and their sum have enough ``weight''.
  The automaton rejects when the total probability of the guessed runs is $\leq \lambda$,
  one of the runs goes into the rejecting sink $q_{rej}$ or a run does not see accepting
  states infinitely often.

  By Lemma~\ref{lem:th_kpba_characterization} we only need to consider sets of runs with
  at most one run that has a probability $<\varepsilon$, where $\varepsilon :=
  \varepsilon_k$ is given by Lemma~\ref{lem:cutoff_fpba_th}. For this single run we also
  do not need to track the exact probability value, as its only purpose is to witness that the
  acceptance probability is strictly greater than $\lambda$, whereas all other runs must have
  one of the finitely many different probabilities which are $\geq \varepsilon$ and must
  sum to $\lambda$.
  \lncsqed
\end{proof}

This generalizes the corresponding result for PFA
\cite[Theorem~3]{fijalkow2017probabilistic}. The proof in
\cite{fijalkow2017probabilistic} uses similar concepts, though a
rather different presentation. In the setting of infinite words we
additionally have to deal with a single run that has arbitrarily low
probability, and we have to ensure that this probability remains
positive.

After seeing that finitely ambiguous PBA retain regularity, we show
that this is the best we can do under threshold semantics:

\begin{corollary}
  \label{cor:nonreg_fpba}
  There are polynomially ambiguous PBA $\mc{A}$, that is, with an
  $\IDA$ pattern and no $\EDA,\IDAF$ patterns, such that
  $L^{>\lambda}(\mc{A})$ is not regular even for rational
  thresholds $\lambda \in ]0,1[$.
\end{corollary}
\begin{proof}
  Follows from the fact that the PWA $\mc{A}$ from Figure~\ref{fig:pwa_fig}(a),
  which recognizes a non-regular language
  (and is used to show Proposition~\ref{prp:th_pwa_nonreg_strong}),
  has just an $\IDA$ pattern in the underlying NBA, but no $\EDA$ or $\IDAF$ patterns.
  \lncsqed
\end{proof}

This completes our characterization of languages which are recognized by PBA that are
restricted by forbidden ambiguity patterns, so that we can state our main result of this
section (see Figure~\ref{fig:hpba_and_amb} for a visualization):

\begin{theorem}
  \label{thm:pbaamb_classes}
  The following results hold about PBA with restricted ambiguity:
  \begin{itemize}
    \item $\mbb{L}^{>0}(k\text{-}\PBA) = \mbb{L}^{>0}(\aleph_0\text{-}\PBA) = \mbb{L}(\NBA)$
    \item $\mbb{L}^{=1}(k\text{-}\PBA) = \mbb{L}^{=1}(2^k\text{-}\PBA) =
      \mbb{L}^{=1}(\text{flat\ }\PBA) = \mbb{L}(\DBA) \subset \mbb{L}^{=1}(\aleph_0 \text{-}\PBA)$
    \item $\mbb{L}^{>\lambda}(k\text{-}\PBA) = \mbb{L}(\NBA) \subset \mbb{L}^{>\lambda}(n^k\text{-}\PBA)$
  \end{itemize}
\end{theorem}
\begin{proof}
  The statements follow from the following inclusion chains:
  \begin{gather*}
    \mbb{L}(\NBA) \overset{(1.)}{\subseteq} \mbb{L}^{>0}(k\text{-}\PBA)
                     \overset{def.}{\subseteq} \mbb{L}^{>0}(\aleph_0\text{-}\PBA)
                     \overset{(2.)}{\subseteq} \mbb{L}(\NBA)
                     \\
    \mbb{L}(\DBA) \overset{(3.)}{\subseteq} \mbb{L}^{=1}(k\text{-}\PBA)
           \overset{def.}{\subseteq}
           \mbb{L}^{=1}(2^k\text{-}\PBA \cup \text{flat\ }\PBA)
\overset{(4.)}{\subseteq} \mbb{L}(\DBA)
           \overset{(5.)}{\subset} \mbb{L}^{=1}(\aleph_0 \text{-}\PBA)
           \\
    \mbb{L}(\NBA) \overset{(1.)}{\subseteq} \mbb{L}^{>0}(k\text{-}\PBA)
           \overset{(6.)}{\subseteq} \mbb{L}^{>\lambda}(k\text{-}\PBA)
           \overset{(7.)}{\subseteq} \mbb{L}(\NBA)
           \overset{(8.)}{\subset} \mbb{L}^{>\lambda}(n^k\text{-}\PBA)
  \end{gather*}
  Where the marked relationships hold due to:
    (1.) Theorem~\ref{thm:nba_to_kpba},
    (2.) Theorem~\ref{thm:ps_cpba_to_nba},
    (3.) Proposition~\ref{prp:dba_to_kpba},
    (4.) Theorem~\ref{thm:as_efpba_to_dba},
    (5.) Proposition~\ref{prp:amb_counterexamples},
    (6.) Simple transformation by adding a new accepting sink $q_{acc}$ and modifying the initial
      distribution $\mu_0$
\cite[Lemma 4.16]{baier2012probabilistic},
    (7.) Theorem~\ref{thm:th_kpba_to_nba},
    (8.) Corollary~\ref{cor:nonreg_fpba}, and
    (def.) by definition of the ambiguity-restricted automata classes.
  \lncsqed
\end{proof}

 \section{Complexity results}
\label{sec:compl}

In this section, we state some upper and lower bounds on the complexity for deciding
emptiness and universality for PBA with restricted ambiguity, derived from the
characterizations and constructions presented above.

\begin{theorem}
  \label{thm:pbaamb_complexity}\ \\[-5mm]
  \begin{enumerate}
    \item the emptiness problem for $\aleph_0$-PBA$^{>0}$ is in $\complclass{NL}$
\item the universality problem for $\aleph_0$-PBA$^{>0}$ is in $\complclass{PSPACE}$
    \item the universality problem for at most exp. ambiguous or flat PBA$^{=1}$ is in $\complclass{NL}$
  \end{enumerate}
\end{theorem}
\begin{proof}
  $(1.+2.):$ By Theorem~\ref{thm:ps_cpba_to_nba} the languages of
  $\aleph_0$-PBA$^{>0}$ are regular. The construction of an NBA
  just uses two copies of the given PBA. For emptiness, it thus
  suffices to guess an accepted ultimately periodic word and verify
  that it is accepted by the NBA, which can be done in NL. Since
  universality for NBA in in PSPACE \cite{SistlaVW85}, we also obtain (2.).

  $(3.)$:
  If the automaton is at most exponentially ambiguous,
  there are only finitely many accepting runs on each word and as we
  know by Lemma~\ref{lem:cntamb_pba_runs} that $w\in L^{=1}(\mc{A})$ iff all runs are
  accepting, it suffices to guess a rejecting run in $\underlyingnba{\mc{A}}$, which
  implies that the ultimately periodic word $w$ labelling that run can not be in
  $L^{=1}(\mc{A})$. If the automaton is flat, then we know that for each rejected word there
  must exist a limit-deterministic rejecting run in the underlying NBA, which we also can
  guess.
  \lncsqed
\end{proof}

\vspace{-5mm}
\begin{table}\begin{center}
\begin{tabular}{c|ccc|cc|cc}
  Type & \multicolumn{3}{c|}{regular?}
       & \multicolumn{2}{c|}{Emptiness}
       & \multicolumn{2}{c}{Universality} \\
       & $>0$ & $=1$ & $>\lambda$ & $>0$ & $=1$ & $>0$ & $=1$ \\\hline

  $k$-$\PBA$ & \multirow{5}{*}{} & \multirow{4}{*}{\hspace{-4mm}\vspace{-4mm}\cmark} &
  & \multirow{3}{*}{$\in\complclass{NL}$}
  & \multirow{3}{*}{$\in\complclass{PSPACE}$}
  & \multirow{3}{*}{$\in\complclass{PSPACE}$}
  & \multirow{3}{*}{$\in\complclass{NL}$} \\\cline{4-4}

  $n^k$-$\PBA$ & & & \multicolumn{1}{|c|}{\multirow{3}{*}{\xmark}} & & & \\
  $2^n$-$\PBA$ & & & \multicolumn{1}{|c|}{} & & & \\\cline{5-8}
  flat $\PBA$  & & & \multicolumn{1}{|c|}{}
& \multirow{2}{*}{$\in\complclass{NL}$ c.}
  & \multirow{2}{*}{$\in\complclass{PSPACE}$ c.}
  & \multirow{2}{*}{$\in\complclass{PSPACE}$ c.}
  & {$\in\complclass{NL}$ c.} \\\cline{3-3}
  $\aleph_0$-$\PBA$ & & \multicolumn{2}{|c|}{\hspace{1cm}\quad}
  & & & & $\in\complclass{PSPACE}$
\end{tabular}
\end{center}
  \caption{Summary of main results from Theorems~\ref{thm:pbaamb_classes} and
  \ref{thm:pbaamb_complexity} concerning PBA with ambiguity restrictions.
  The completeness results follow from the hardness results for HPBA (which are subsumed by flat PBA)
  from \cite[Section 5]{chadha2011randomization}, the $\complclass{PSPACE}$ inclusion of
  universality for almost-sure $\aleph_0$-$\PBA$ follows from \cite[Theorem 4.4]{chadha2011randomization}.
  \vspace{-5mm}
  }
  \label{tab:result_overview}
\end{table}

Observe that $\aleph_0$-PBA$^{>0}$ subsume HPBA$^{>0}$ and the union of flat PBA$^{=1}$ and exp.
ambiguous PBA$^{=1}$ subsumes HPBA$^{=1}$, while preserving the same complexity of
the emptiness and universality problems. A summary of the main results from
Theorem~\ref{thm:pbaamb_classes} and Theorem~\ref{thm:pbaamb_complexity} is presented in
Table~\ref{tab:result_overview}.

We conclude with an observation relevant to the question about feasibility of PBA with restricted
ambiguity for the purpose of application in e.g. model-checking or synthesis.
\begin{proposition}[Relationship to classical formalisms]
  \
  \begin{itemize}
    \item There is a doubly-exponential lower bound for translation from LTL formula to
      countably ambiguous PBA with positive semantics.
    \item There is an exponential lower bound for conversion from NBA to countably
      ambiguous PBA with positive semantics.
  \end{itemize}
\end{proposition}
\begin{proof}
  It is known \cite[Theorem 2]{esparza2016ltlldba} that there is a doubly-exponential lower bound
    from LTL to LDBA. It is also known that LTL to NBA has an exponential lower bound (e.g. \cite[Theorem
    5.42]{baier2008principles}), which implies an exponential lower bound from NBA to
    LDBA.

    By Theorem~\ref{thm:ps_cpba_to_nba} there is a polynomial transformation from countably
    ambiguous PBA with positive semantics into LDBA, which together with the aforementioned
    bounds implies the claimed lower bounds.
  \lncsqed
\end{proof}
 \section{Weakness in Probabilistic Büchi Automata}
\label{sec:weak}

In this section we investigate the class of probabilistic weak
automata (PWA), establishing the relation between different classes
defined by PWA as shown in Figure~\ref{fig:illustration_weak_classes}
(see also the description of our contribution in the introduction).

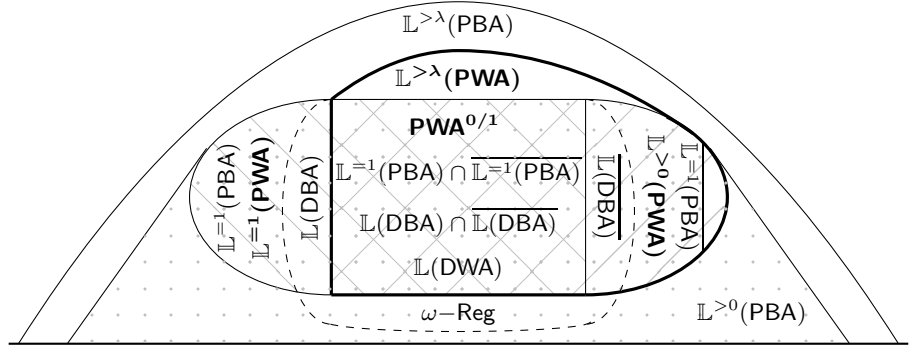
\begin{figure}[htbp]
  \centering
  
\begin{tikzpicture}
\pgftransformscale{.65}

  \newcommand{\xmin}{-2.6}
  \newcommand{\xmax}{2.6}
  \newcommand{\ymin}{1}
  \newcommand{\ymax}{5}

\draw[thick] (9.2,0) -- (-9.2,0);

\draw[] (-9,0) parabola bend (0,7) (9,0);
  \node at (0,6.5) { \begin{tabular}{c} $\mathbb{L}^{>\lambda}(\PBA)$ \end{tabular} };

\draw[very thick] (\xmin,\ymax) parabola bend (0,6) (\xmax+2.2,\ymax-0.7) ;
\node at (0,5.4) { \begin{tabular}{c}{\boldmath$\mathbb{L}^{>\lambda}(\PWA)$} \end{tabular} };

\draw[pattern color=gray!50, pattern=northwest, hatch distance=25pt] (\xmin,\ymax) arc(90:270:2.9cm and 2cm);
\fill[pattern color=gray!50, pattern=mydots] (\xmin,\ymax) arc(90:270:2.9cm and 2cm);
\node[rotate=90] at (-4.7,3) { \begin{tabular}{c} ${\mathbb{L}^{=1}(\PBA)}$ \end{tabular} };
\node[rotate=90] at (-4,3) { \begin{tabular}{c} {\boldmath$\mathbb{L}^{=1}(\PWA)$} \end{tabular} };

\draw[pattern color=gray!50, pattern=northeast, hatch distance=25pt] (\xmax,\ymax) arc(90:-90:2.9cm and 2cm);
\draw[very thick,pattern color=gray!50, pattern=northeast, hatch distance=25pt] (\xmax,\ymin) arc(-90:55:2.9cm and 2cm);
\fill[pattern color=gray!50, pattern=mydots] (\xmax,\ymax) arc(90:-90:2.9cm and 2cm);
\node[rotate=-90] at (4.7,3) { \begin{tabular}{c} $\overline{\mathbb{L}^{=1}(\PBA)}$ \end{tabular} };
\node[rotate=-90] at (4,3) { \begin{tabular}{c} {\boldmath${\mathbb{L}^{>0}(\PWA)}$} \end{tabular} };

\fill [pattern color=gray!50, pattern=northwest, hatch distance=25pt] (\xmin,\ymax) rectangle (\xmax,\ymin);
\fill [pattern color=gray!50, pattern=northeast, hatch distance=25pt] (\xmin,\ymax) rectangle (\xmax,\ymin);
\fill [pattern color=gray!50, pattern=mydots] (\xmin,\ymax) rectangle (\xmax,\ymin);
\draw[] (\xmin,\ymax) -- (\xmax,\ymax);
\draw[very thick] (\xmin,\ymax) -- (\xmin,\ymin);
\draw[very thick] (\xmin,\ymin) -- (\xmax,\ymin);
\draw[] (\xmax,\ymin) -- (\xmax,\ymax);
\node at (0,4.4) { \begin{tabular}{c} {\boldmath $\PWA^{0/1}$ \unboldmath} \end{tabular} };
\node at (0,3.5) { \begin{tabular}{c} $\mathbb{L}^{=1}(\PBA) \cap \overline{\mathbb{L}^{=1}(\PBA)}$ \end{tabular} };

\draw[] (-5.15,4) -- (-8,0);
\draw[] (5.15,4) -- (8,0);

\fill[pattern color=gray!50, pattern=mydots] (-8,0) -- (-5.15,4) -- (5.15,4) -- (8,0);

\node at (6,0.6) { \begin{tabular}{c} $\mathbb{L}^{>0}(\PBA)$ \end{tabular} };

\draw[dashed] (\xmin,\ymax) arc(90:270:1cm and 2.3cm);
\draw[dashed] (\xmax,\ymax) arc(90:-90:1cm and 2.3cm);
\draw[dashed] (2.7,0.45) arc(0:-180:2.7cm and 0.2cm);
\node at (0,0.6) { \begin{tabular}{c} $\op{\omega}{-}\op{Reg}$ \end{tabular} };

  \node[rotate=-90] at (3,3) { \begin{tabular}{c} $\overline{\mbb{L}(\DBA)}$ \end{tabular} };
  \node[rotate=90] at (-3,3) { \begin{tabular}{c} ${\mbb{L}(\DBA)}$ \end{tabular} };
\node at (0,2) { \begin{tabular}{c}
  $\mbb{L}(\DBA)\cap\overline{\mbb{L}(\DBA)}$ \\[2mm]
  $\mbb{L}(\op{DWA})$
  \end{tabular} };
\let\xmin\undefined
\let\xmax\undefined
\let\ymin\undefined
\let\ymax\undefined
\end{tikzpicture}
   \caption{Illustration of relationships between the class of languages accepted by weak
  probabilistic automata under
  various semantics with other already known classes. The overlapping patterns indicate intersection
  of classes, where dots mark $\mbb{L}^{>0}(\PBA)$, and different diagonal lines
  respectively $\mbb{L}^{=1}(\PBA)$ and $\overline{\mbb{L}^{=1}(\PBA)}$.
  The dashed line indicates intersections with different subclasses of
  regular languages. The class $\mathbb{L}^{>\lambda}(\PBA)$ contains all the
  other depicted classes, $\mathbb{L}^{>\lambda}(\PWA)$ contains the
  area inside the thick line.
The depicted fact that
  $\mbb{L}^{>0}(\PWA) = \mathbb{L}^{>\lambda}(\PWA) \cap \mathbb{L}^{>0}(\PBA)$ is
  a conjecture, one direction is shown in Theorem~\ref{thm:expressive_power_pwa}.
  }
  \label{fig:illustration_weak_classes}
\end{figure}

As a first remark, notice that PWA can be ``complemented'' by
inverting accepting and rejecting states and switching between dual
semantics, e.g., for a PWA $\mc{A}$ we have
$\overline{L^{>0}(\mc{A})} = L^{=1}(\overline{\mc{A}})$, where
$\overline{\mc{A}}$ is just $\mc{A}$ with inverted accepting state set
$F' = Q \setminus F$.

Since the overarching theme of this paper is trying to find regular
subclasses of PBA, we will next establish the following result, showing
that there is no hope to find a complete syntactical characterization
of regularity in PBA:

\begin{restatable}{theorem}{thmregundec}
  \label{thm:reg_undec}
  The regularity of PWA (and therefore of PBA) under positive, almost-sure and threshold semantics
  is an undecidable problem.
\end{restatable}
\begin{proof}[sketch] Since $\mbb{L}^{>\lambda}(\text{PWA}) \supseteq
  \mbb{L}^{>0}(\text{PWA})$ (see Theorem~\ref{thm:expressive_power_pwa}),
  $\mbb{L}^{>0}(\text{PWA}) = \overline{\mbb{L}^{=1}(\text{PWA})}$,
  and the class of regular $\omega$-languages is closed under
  complement, it suffices to show the statement for PWA$^{=1}$.
We do this by reduction from the value 1
  problem for PFA, which is the question whether for each $\varepsilon>0$ there exists a
  word accepted by the PFA with probability $> 1-\varepsilon$. This problem is known to be
  undecidable \cite{gimbert2010probabilistic}. We consider a
  slightly modified version of the problem by assuming that no word
  is accepted with probability 1 by the given PFA. The problem remains
  undecidable under this assumption, because one can check if a
  PFA accepts a finite word with probability 1 by a simple subset construction.

  Given some PFA $\mc{A}$, we construct a PWA$^{=1}$ $\mc{B}$ by
  taking a copy of $\mc{A}$ and extending it with a new symbol $\#$ such that from
  accepting states of $\mc{A}$ the automaton is ``restarted'' on $\#$, while from
  non-accepting states $\#$ leads into a new part which ensures that infinitely many
  $\#$ are seen and contains the only accepting state of $\mc{B}$. We show that
  $L^{=1}(\mc{B}) = (\Sigma^*\#)^\omega \setminus R$, where $R=\emptyset$ if $\mc{A}$ does
  not have value 1, and $R$ is non-empty but  does not contain an ultimately
  periodic word, otherwise. This implies that $L^{=1}(\mc{B})$ is
  regular iff $\mc{A}$ does
  not have value 1. \lncsqed
\end{proof}

We will now show that PWA with almost-sure semantics are as expressive
as PBA, and with positive semantics as expressive as PCA.
\begin{restatable}{theorem}{thmpcatopwa}
  \label{thm:pca_to_pwa}
  $\mathbb{L}^{>0}(\PWA) = \mathbb{L}^{>0}(\PCoBA)$ and
  $\mathbb{L}^{=1}(\PWA) = \mathbb{L}^{=1}(\PBA)$.
\end{restatable}
\begin{proof}[sketch] It suffices to show the first statement. The second then follows by duality, i.e., we can
  interpret a PBA$^{=1}$ $\mc{A}$ recognizing $L$ as a PCA$^{>0}$ recognizing
  $\overline{L}$ and just apply the construction to get a PWA$^{>0}$ $\mc{B}$ for $\overline{L}$,
  such that $\overline{\mc{B}}$ (with inverted accepting and rejecting states) is a
  PWA$^{=1}$ for $L$. In the first statement the $\subseteq$ inclusion is trivial, hence
  we only need to show that $\mathbb{L}^{>0}(\PCoBA) \subseteq \mathbb{L}^{>0}(\PWA)$.

  We construct a PWA$^{>0}$ consisting of two copies of the original PCA$^{>0}$,
  a \emph{guess} copy and a \emph{verify} copy. In the first copy, the automaton can guess
  that no final states will be visited anymore and switch to the verify copy, which is
  accepting, but where all transitions into final states are redirected to a rejecting
  sink.
  \lncsqed
\end{proof}

Next, we show that languages that can be accepted by both, a PWA with almost-sure
semantics, and by a PWA with positive semantics, are regular and
can be accepted by a DWA. For the proof, we rely on a characterization
of DWA languages in terms of the Myhill-Nerode equivalence relation
from \cite{staiger1983finite}. So we first define this equivalence,
and show that languages defined by PBA with positive semantics have
only finitely many equivalence classes. Then we come back to the
result for PWA.

For $L\subseteq \Sigma^\omega$, define the Myhill-Nerode equivalence relation
$\sim_L \subseteq \Sigma^* \times \Sigma^*$ by $u \sim_L v$ iff $uw\in
L \Leftrightarrow vw\in L$ for all $w\in\Sigma^\omega$.  Then the
following holds:
\begin{lemma}[Finitely many Myhill-Nerode classes]
  \label{lem:fin_myhnerode_classes}\ \\
  Languages in $\mbb{L}^{>0}(\PBA)$ have finitely many Myhill-Nerode equivalence classes.
\end{lemma}
\begin{proof}
  Let $\mc{A}=(Q,\Sigma,\delta,\mu_0,F)$ be some PBA$^{>0}$ and $u\in\Sigma^*$ some word
  and let $\mu_u := \delta^*(\mu_0, u)$ be the probability distribution on states of $\mc{A}$
  after reading $u$. Pick any $w\in \Sigma^\omega$ and notice that $uw \in L=L^{>0}(\mc{A})$
  iff there exists some state $q$ such that $\mu_u(q) > 0$ and the probability to accept
  $w$ from $q$ is also $>0$, as the product of two positive numbers clearly still
  is positive. But then, for any two $u,v \in \Sigma^*$ we have that whenever $\mu_u(q) >0
  \Leftrightarrow \mu_v(q) > 0$ for all $q$, then we have $uw \in L \Leftrightarrow vw \in
  L$ for all $w\in \Sigma^\omega$ by the reasoning above, as the exact value does not
  matter for acceptance, and therefore $u \sim_L v$. But as
  there are only at most $2^{|Q|}$ different possibilities how values in a distribution
  $\mu$ over $Q$ are either equal to or greater than $0$, this is an upper bound on
  the number of different equivalence classes.
  \lncsqed
\end{proof}

\begin{theorem} \label{the:pwa_extremal}
  $\mbb{L}^{>0}(\PWA) \cap \mbb{L}^{=1}(\PWA) = \mbb{L}(\DWA) = \mbb{L}(\PWA^{0/1})$
\end{theorem}
\begin{proof}
The inclusions $\mbb{L}(\DWA) \subseteq \mbb{L}(\PWA^{0/1})
\subseteq \mbb{L}^{>0}(\PWA) \cap \mbb{L}^{=1}(\PWA)$ are trivial,
hence it remains to show $\mbb{L}^{>0}(\PWA) \cap \mbb{L}^{=1}(\PWA)
\subseteq \mbb{L}(\DWA)$.

So let $L$ be a language from $\mbb{L}^{>0}(\PWA) \cap
\mbb{L}^{=1}(\PWA)$. We want to show that $L$ can be accepted by a
DWA.  We use the following characterization of DWA languages
\cite[Theorem 21]{staiger1983finite}: The DWA languages are precisely
the languages with finitely many Myhill-Nerode classes in the class
$G_\delta \cap F_\sigma$ in the Borel hierarchy. The classes
$G_\delta$ and $F_\sigma$ of the Borel hierarchy are often also
referred to as $\Pi_2$ and $\Sigma_2$. We do not introduce the details
of this hierarchy here, but rather refer the reader not familiar with
these concepts to \cite{staiger1983finite} and
\cite{chadha2011randomization}.

We already know that $L$ has finitely many Myhill-Nerode classes by
Lemma~\ref{lem:fin_myhnerode_classes} (as PWA are special cases of
PBA). It remains to show that $L$ is in the class $G_\delta \cap
F_\sigma$. It is known that PBA with almost-sure semantics define
languages in $G_\delta$ \cite[Lemma
  3.2]{chadha2011randomization}. Hence $L$ is in $G_\delta$. Since $L$
is accepted by a PWA with positive semantics, the complement of $L$ is
accepted by a PWA with almost-sure semantics (as noted at the
beginning of this section). We obtain that the complement of $L$ is
also in $G_\delta$ again by \cite[Lemma
  3.2]{chadha2011randomization}. This means that $L$ is in $F_\sigma$,
which by definition consists of the complements of languages from
$G_\delta$.
\lncsqed
\end{proof}

Concluding this section, we show a result about weak automata with threshold
semantics, which (not surprisingly) turn out to be even more expressive.
A careful analysis
of the PWA $\mc{A}$ in Fig.~\ref{fig:pwa_fig}(a)
shows the following result:
\begin{restatable}{proposition}{prpthpwanonreg}
  \label{prp:th_pwa_nonreg_strong}
For all thresholds $\lambda \in ]0,1[$ there exists
  a PWA $\mc{A}$ such that
  $L^{>\lambda}(\mc{A})$ is not regular and not $PBA^{>0}$ recognizable.
\end{restatable}

Putting things together, we can say the following about threshold PWA,
establishing the relation of $\mathbb{L}^{>\lambda}(\PWA)$ to the other classes in
Figure~\ref{fig:illustration_weak_classes}:
\begin{theorem}[Expressive power of threshold PWA]
  \label{thm:expressive_power_pwa}
  \quad\\[-5mm]
  \begin{enumerate}
    \item $\mbb{L}^{>0}(\PWA) \subseteq \mathbb{L}^{>\lambda}(\PWA) \cap \mathbb{L}^{>0}(\PBA)$.
    \item $\mathbb{L}^{>\lambda}(\PWA)$ and $\mathbb{L}^{>0}(\PBA)$ are incomparable (wrt. set inclusion).
    \item $\mathbb{L}^{>0}(\PWA) \subset \mathbb{L}^{>\lambda}(\PWA) \subset \mathbb{L}^{>\lambda}(\PBA)$.
\end{enumerate}
\end{theorem}
\begin{proof}
  (1.) $\mbb{L}^{>0}(\PWA) \subseteq \mathbb{L}^{>0}(\PBA)$ by definition and
  $\mbb{L}^{>0}(\PWA) \subseteq \mathbb{L}^{>\lambda}(\PWA)$, as any PWA$^{>0}$ can be
  modified to a PWA$^{>\lambda}$ recognizing the same language by just adding an
  additional accepting sink and modifying the initial distribution, just as described in
  \cite[Lemma 4.16]{baier2012probabilistic} for general PBA.

  (2.) By Proposition~\ref{prp:th_pwa_nonreg_strong}, there are
  languages recognized by PWA$^{>\lambda}$ that cannot be recognized with PBA$^{>0}$.
To show that there are languages accepted by PBA$^{>0}$ that cannot be accepted by
  PWA$^{>\lambda}$ we can give a topological characterization of languages accepted by
  PWA by a simple adaptation of \cite[Lemma 3.2]{chadha2011randomization} and combine it
  with other results shown in \cite{chadha2011randomization} to show that there are
  PBA$^{>0}$ that accept languages that cannot be accepted by PWA$^{>\lambda}$.

  (3.) The first inclusion was discussed in (1.),
  the strictness follows from Proposition~\ref{prp:th_pwa_nonreg_strong} and the fact that
  $\mathbb{L}^{>0}(\PWA) = \overline{\mbb{L}^{=1}(\PBA)} \subset
  \op{BCl}(\mbb{L}^{=1}(\PBA)) = \mathbb{L}^{>0}(\PBA)$, where the first equality is
  Theorem~\ref{thm:pca_to_pwa} and the second is shown in \cite{chadha2011randomization}.
  The second inclusion of the statement follows from (2.) and the fact from \cite{baier2012probabilistic}
  that $\mbb{L}^{>0}(\PBA) \subset \mbb{L}^{>\lambda}(\PBA)$.
  \lncsqed
\end{proof}

For the dual class $\mbb{L}^{\geq \lambda}(\PWA)$ one can show symmetric results that
correspond to statements (1.) and (2.) above, for statement (3.) however there is no proof
yet for the strictness of the inclusions (especially the second one), whereas the
statement $\mathbb{L}^{=1}(\PWA) \subseteq \mathbb{L}^{\geq\lambda}(\PWA) \subseteq
\mathbb{L}^{\geq\lambda}(\PBA)$ is obvious. We leave this issue as an open
question. Another interesting question is whether $>\lambda$ is equivalent to $<\lambda$
(or dually for $\geq / \leq$).

 \section{Conclusion}
\label{sec:conclusion}

By using
notions from ambiguity in classical Büchi automata, we were able to extend the set of easily
(syntactically) checkable PBA which are regular under some or all of the usual
semantics. As a consequence, ambiguity appears to be an even more interesting notion in
the probabilistic setting, as here it in fact has consequences for the expressive power of
automata, whereas in the classical setting there is no such effect. Our results also
indicate that to get non-regularity, one requires the use of certain structural
patterns which at least imply the existence of the ambiguity patterns that we used. It
is an open question whether it is possible to identify more fine-grained syntactic
characterizations, patterns or easily checkable properties which are just over-approximated by the
ambiguity patterns and are required for non-regularity.

\bibliographystyle{splncs04}
\bibliography{literature}

\clearpage
\appendix
\section{Proofs for section on ambiguity in PBA}
\label{app:amb}

\subsection{Proof for Proposition~\ref{prp:hpba_ambiguity}}
\prphpbaambiguity*
\begin{proof}
  \
  \begin{enumerate}
    \item
      First, observe that all states in the same SCC of a HPBA must have the same rank, as
      otherwise the SCC contains a path where the ranks of the states strictly decrease.
      Existence of an EDA pattern implies that there is at least one intra-SCC fork, which
      implies that two successors must have the same rank, which is forbidden for HPBA.

      On the other hand, it is easy to construct an automaton that has no EDA pattern, but
      is not a valid HPBA because it has an intra-SCC fork. The second inclusion follows
      trivially because there are automata with EDA pattern but no EDA$_F$ pattern, and
      thus are at most countably ambiguous (by \cite{faba-dlt}).

    \item Finitely ambiguous automata have no IDA (and thus no EDA) patterns
      (by \cite{faba-dlt}), but even unambiguous PBA may contain an intra-SCC
      fork, meaning that it cannot be a HPBA. On the other hand, HPBA may even have an
      IDA$_F$ pattern, which then implies infinite ambiguity.

    \item Clearly, level 1 of SPBA can be thought of a rejecting sink, as no accepting
      states are reachable. As in the trimmed automaton there is just one (useful) SCC containing
      states on level 0 and there are no intra-SCC forks in HPBA, transitions within level
      0 are deterministic. Hence SPBA are trivially unambiguous, which by definition is a
      strict subset of finitely ambiguous PBA.
  \end{enumerate}
  \lncsqed
\end{proof}

\clearpage

\subsection{Proof for Theorem~\ref{thm:nba_to_kpba}}
\thmnbatokpba*
\begin{proof}
  Let $\mc{A}=(Q,\Sigma,\delta,q_0,c)$ be a deterministic parity automaton accepting $L$,
  i.e., a finite automaton with priority function $c:Q\to\{1,\ldots,m\}$ such that $w \in
  L(\mc{A})$ iff the smallest priority assigned to a state on the unique run of $\mc{A}$
  on $w$ which is seen infinitely often is even.

  We will construct an unambiguous LDBA $\mc{A}'$ from $\mc{A}$ which also accepts $L$,
  from which we will easily obtain an unambiguous PBA $\mc{B}$.
  For this, we take $m+1$ copies of $\mc{A}$ and create a Büchi automaton which guesses the
  smallest priority that is seen infinitely often along the run in $\mc{A}$, and ensure
  that only one correct guess is possible for each word.

  Formally, let $\mc{A}'=(Q', \Sigma, \Delta', Q_0', F')$ be an NBA with
  $Q' := \tilde{Q}\cupdot Q_1 \cupdot \ldots \cupdot Q_m$ consisting of $m+1$
  copies of each state in $Q$, where the copies of $q\in Q$ are denoted by $\tilde{q}, q^1,
  \ldots, q^m$, respectively, initial states defined as
  $Q_0' := \{\tilde{q_0}, q_0^1, \ldots, q_0^m \}$ and final states defined as
  $F' := \{ q^i \mid c(q) = i \text{\ and\ } i \text{\ is\ even}\}$.
  The transition relation $\Delta' := \tilde{\Delta} \cup \bigcup_{i=1}^m \Delta_m$
  is given by
  \begin{itemize}
    \item $\tilde{\Delta} := \{ (\tilde{p},a,q') \mid \delta(p,a)=q \text{\ and\ }
      (q'=\tilde{q} \text{\ or\ } q'=q^j \text{\ s.t.\ } c(q)\geq j > c(p)) \}$,
    \item $\Delta_i := \{ (p^i, a, q^i) \mid \delta(p,a)=q \text{\ and\ } c(p), c(q)\geq i \}$.
  \end{itemize}
  As the transitions defined by the $\Delta_i$ sets are just copies of a subset of the
  deterministic transitions given by $\delta$, and as all accepting states are only in
  these restricted deterministic copies of $\mc{A}$, clearly $\mc{A}'$ is an LDBA. Now we
  will show that it accepts the same language and is unambiguous.

  If $w \in L(\mc{A}')$, then there exists a run $\rho$ which either reaches or starts in
  one of the $m$ copies of $\mc{A}$ that contain accepting states, and visits those infinitely
  often. Notice that we can easily obtain the run of $\mc{A}$ on $w$ by projecting the
  states of $\rho$ onto the original states in $\mc{A}$. Now w.l.o.g. assume that $\rho$
  eventually is in the $i$-th copy, i.e., eventually using states in $Q_i$. As $\rho$ is
  accepting, we have by definition of $F'$ that $i$ must be even and $\rho$ visits states
  $q^i$ in $\mc{A}'$ such that $c(q) = i$ in $\mc{A}$ infinitely often. Also, $\rho$
  eventually never visits states $q^i$ with $c(q)<i$ in $\mc{A}$, as transitions with such
  states are not defined in $\Delta_i$. This implies that the run of $\mc{A}$ on $w$ is accepting.

  If $w \in L(\mc{A})$, let $\rho$ now be the run on $w$ in $\mc{A}$,
  $k$ the minimal priority $k$ which is seen along $\rho$ infinitely often, and $t$ the
  time where a state with priority $<k$ is visited for the last time (or if $\rho$ never
  visits such states, let $t:=-1$).

  First consider the runs that start in some initial state $q_0^j$, which all proceed
  deterministically in the corresponding restricted $j$-th copy of $\mc{A}$. If $j>k$,
  then at some point when in $\rho$ a state $q$ with priority $\leq k$ is visited, there
  exists no matching transition to $q^j$, so the run from $q_0^j$ terminates. If $j<k$, then if
  $\rho$ reaches a state with priority $<j$, the run also terminates, and otherwise at
  some point $\rho$ does not see states with priority $<k$, so that by definition of $F'$
  the run does not see accepting states anymore, and hence the run is rejecting. If $j=k$,
  then the run terminates if some state with priority $<k$ is visited at some point along
  the run $\rho$, and otherwise (the case with $t=-1$) it can continue forever.
  Furthermore, by choice of $k$, states $q^j$ with $c(q)=k$ are visited infinitely often,
  so that by definition of $F'$ the run is accepting.

  Now consider the runs which start in $\tilde{q}_0$ and observe that the automaton can
  either use the unique transitions between states in $\tilde{Q}$, or at any point
  nondeterministically decide to switch into one of the restricted copies discussed above,
  but from any state $\tilde{p}$ only to a $q^j \in Q_j$  in a copy of $\mc{A}$ where only
  copies of states $q\in Q$ with priorities $c(q) \geq j > c(p)$ can be reached.

  If $t=-1$, i.e., no state with priority $<k$ is ever visited by $\rho$, the runs of $\mc{A}'$
  which forever visit states in $\tilde{Q}$ are all rejecting, whereas runs that
  eventually switch into one of the other copies can only choose to go to a copy with
  states $Q_j$ with $j>k$ and hence these runs must terminate whenever $\rho$ visits a
  state $q$ with $c(q)=k$, which happens infinitely often, so all runs from $\tilde{q}_0$
  are rejecting.

  For $t\geq 0$, observe the following.
  If a run eventually switches from $\tilde{Q}$ to some state in $Q_j$ with
  $j\neq k$, then as discussed above the run will either terminate (due to missing
  transitions in $\Delta_j$) or be rejecting (by definition of $F'$). Furthermore, if it
  switches too early to a state in $Q_k$, it will also terminate (as $\rho$ will visit at
  least one more state with priority $<k$), and a run cannot switch to states in $Q_k$
  strictly after $t$, because by definition of $\tilde{\Delta}$ this is only possible from
  a state with priority $<k$.
Hence, the only possible accepting run is the one which stays in $\tilde{Q}$
  until time $t$ and in the next transition switches to some state $q^k \in Q_k$, from
  where it continues deterministically and accepts, as then no more states with priority $<k$
  are visited by $\rho$ and hence no transitions that are missing in $\Delta_k$ are used, and
  furthermore infinitely many states $q^k\in F'$ are visited, which are copies of states $q$
  with $c(q)=k$.

  So in any case, for every accepting run $\rho$ in $\mc{A}$ there exists exactly one
  accepting run in $\mc{A}'$: for $t=-1$ it is the run starting in $q_0^k$, and for $t\geq
  0$ it is the run starting in $\tilde{q}$ and switching to a state in $Q_k$ in the
  transition from time $t$ to $t+1$.
Therefore $\mc{A}'$ is an unambiguous LDBA accepting $L$. As all
  accepting runs in $\mc{A}'$ are limit-deterministic, we can trivially obtain the claimed
  unambiguous PBA $\mc{B}$ which accepts $L$ under positive semantics by equipping edges
  in $\mc{A}'$ with arbitrary probabilities that result in valid probability
  distributions, because in any case the unique limit-deterministic accepting runs in
  $\mc{B}$ will have positive probability.
  \lncsqed

\end{proof}

\clearpage

\subsection{Proof for Theorem~\ref{thm:ps_cpba_to_nba}}
\thmpscpbatonba*
\begin{proof}
  Let $\mc{A}=(Q,\Sigma,\delta,\mu_0,F)$ be a PBA that is at most countably ambiguous. We construct an NBA $\mc{B}$ accepting $L^{>0}(\mc{A})$, which intuitively consists of
  two copies of $\underlyingnba{\mc{A}}$. The first copy has no accepting states
  and the second copy has no forks.

  Let $\mc{B}=(Q', \Sigma, \Delta', Q'_0, F')$ be an NBA,
  where $Q' = Q \times \{n,d\}$ consists of two copies of each state in $\mc{A}$,
  $Q'_0 = \{ (q,n) \mid \mu_0(q)>0 \}, F' = \{ (q,d) \mid q \in F \}$,
  and transitions $\Delta := \Delta_n \cupdot \Delta_d \cupdot \Delta_{nd}$ defined by
  \begin{itemize}
    \item $\Delta_n    = \{ ((p,n),a,(q,n)) \mid \delta(p,a,q) > 0 \}$,
    \item $\Delta_{nd} = \{ ((p,n),a,(q,d)) \mid \delta(p,a,q) > 0 \}$, and
    \item $\Delta_d    = \{ ((p,d),a,(q,d)) \mid \delta(p,a,q) = 1 \}$.
  \end{itemize}

  It is easy to see that the automaton accepts exactly those words for which there
  exists a limit-deterministic accepting run, hence
  by Lemma~\ref{lem:cntamb_pba_runs} we have $L^{>0}(\mc{A}) = L(\mc{B})$.
  \lncsqed
\end{proof}

\subsection{Proof for Theorem~\ref{thm:as_efpba_to_dba}}
\thmasefpbatodba*
\begin{proof}
  Let $\mc{A}=(Q,\Sigma,\delta,\mu_0,F)$ be a PBA.
  There are two cases to consider---when $\mc{A}$ is exponentially ambiguous and when
  $\mc{A}$ is flat.

  First, assume that $\mc{A}$ is at most exponentially ambiguous, which
  means that on each word there are only finitely many accepting runs.
  We construct a DBA $\mc{B}$ accepting $L^{=1}(\mc{A})$. By
  Lemma~\ref{lem:cntamb_pba_runs}, $\mc{B}$ should accept if \emph{every} run of $\mc{A}$
  accepts and is limit-deterministic.
Notice, that we do not even need to check that the runs on $w$ are limit-deterministic,
  because if all runs accept, this already implies $w\in L^{=1}(\mc{A})$.
Hence, we just need to check that all runs accept, using a simple breakpoint construction.

  Formally, let $\mc{B} := (Q', \Sigma, \delta', q_0', F')$ with $Q' := 2^Q \times 2^Q, q_0'
  := (\emptyset, \op{supp}(\mu_0)),$
  $F' := \{ (S, \emptyset) \mid S \subseteq Q \}$ and transition function $\delta'$ defined by
  \begin{itemize}
    \item $\delta'((S,\emptyset),a) := (\emptyset, \Delta(S,a))$, and
    \item $\delta'((S,T),a) := (S', T')$ for $T\neq\emptyset$ \\
      with $T' = \Delta(T,a)\setminus F$ and $S' := \Delta(S\cup T, a) \setminus T'$.
  \end{itemize}

  It is easy to see that $\mc{B}$ sees accepting states infinitely often if and only if on
  every path in $\mc{A}$ an accepting state is visited infinitely often, and hence
  by Lemma~\ref{lem:cntamb_pba_runs} we have $L^{=1}(\mc{A}) = L(\mc{B})$.

  Now assume that $\mc{A}$ is flat.
  In this case, we construct a DBA $\mc{B}$ accepting $L^{=1}(\mc{A})$, that by
  Lemma~\ref{lem:cntamb_pba_runs} should accept $w$ iff there exists no
  limit-deterministic rejecting run of $\mc{A}$. This is checked using a
  construction almost as above, but now it suffices for a state to be
  at some point reached only by branching transitions to be moved into the left set.

Formally, define $\mc{B}$ as above, but with different $\delta'$ defined by
  \begin{itemize}
    \item $\delta'((S,\emptyset),a) := (\emptyset, \Delta(S,a))$, and
    \item $\delta'((S,T),a) = (S', T')$ for $T\neq\emptyset$ with
      \begin{itemize}
        \item $T' := \{q \mid q\not\in F \text{\ and \ } \exists p\in T \text{\ s.t.\ }\delta(p,a,q) = 1 \}$, and
        \item $S' := \Delta(S\cup T, a) \setminus T'$.
      \end{itemize}
  \end{itemize}

  Let $w=w_0w_1\ldots \in \Sigma^\omega$.
  If $w\not\in L^{=1}(\mc{A})$, by Lemma~\ref{lem:cntamb_pba_runs} there exists a
  limit-deterministic rejecting run $\rho = q_0, q_1, \ldots$ on $w$, then from some time
  $t$ on only deterministic transitions (i.e., with $\delta(q_i, w_i, q_{i+1}) = 1$) will
  be taken and all states $q_i$ for $i\geq t$ are rejecting. Hence by construction the
  set in the right component of the macrostate will always contain the current state along
  the run and thus will never become empty anymore, so no accepting states of $\mc{B}$ are
  visited anymore and hence $w\not\in L(\mc{B})$.

  On the other hand, if $w\in L^{=1}(\mc{A})$, then there are no limit-deterministic
  rejecting runs, which means that every run either sees accepting states infinitely often
  (in which case it is accepting), or uses branching transitions infinitely often (in
  which case it is not limit-deterministic). But then by construction, infinitely often all
  successor states in the sets will reach the left set and the right set must become
  empty, and therefore $w \in L(\mc{B})$.
  \lncsqed
\end{proof}

\subsection{Omitted details for Proposition~\ref{prp:amb_counterexamples}(2)}

\begin{lemma}\ \\
  The automata in Figure~\ref{fig:pwa_fig}(c) accept non-regular languages for
  all $\lambda\in ]0,1[$.
\end{lemma}
\begin{proof}
The PWA presented in Figure~\ref{fig:pwa_fig}(c) is based on the PBA depicted in
\cite[Fig. 6]{baier2012probabilistic} and accepts for some $\lambda\in]0,1[$ the
following language, which is known to be not regular:
\[ \tilde{L}_\lambda = \left\{ a^{k_1}ba^{k_2}b\ldots \mid k_1,k_2,\ldots \in \mathbb{N}_{\geq 1}
\text{\ such\ that\ } \prod_{i=1}^\infty \left(1-(1-\lambda)^{k_i}\right) = 0 \right\} \]

Notice that $a^\omega$ is not accepted, as then $q_f$ can never be
reached. Also, if there are finitely many $b$'s, i.e., the word has the shape
$w=a^{k_1}b \ldots a^{k_n} b a^\omega$, then there is positive probability to
not reach $q_f$ after reading the last $b$ and after that $q_f$ cannot be reached
anymore, hence with positive probability the automaton rejects $w$. Hence it is easy to
see that all accepted words must be of the form $(a^{+}b)^\omega$.

Once a run has reached $q_f$, it becomes accepting and stays accepting forever. The
probability to reach $q_f$ from $q_0$ on $a^k b$ is $(1-\lambda)^k$, whereas the
probability to avoid $q_f$ and come back to $q_0$ instead is $1-(1-\lambda)^k$. Hence,
$\prod_{i=1}^\infty (1-(1-\lambda)^{k_i})$ is the probability of runs that avoid $q_f$
forever and therefore is exactly the probability of rejecting runs. Therefore, we have
$L^{=1}(\tilde{\mc{P}}_\lambda) = \tilde{L}_\lambda$, as claimed.
\lncsqed
\end{proof}

\clearpage

\subsection{Proof for Theorem~\ref{thm:th_kpba_to_nba}}
\thmthkpbatonba*
\begin{proof}
  We use the characterization of Lemma~\ref{lem:th_kpba_characterization} to
  construct a generalized Büchi automaton $\mc{B}$ (i.e., a Büchi automaton with multiple
  acceptance sets, where from each set at least one state must be visited infinitely
  often) accepting $L^{>\lambda}(\mc{A})$, which can easily be translated into an NBA.

  Intuitively, the new automaton $\mc{B}$ just guesses at most $k$ different runs of
  $\mc{A}$ and verifies that the guessed runs are limit-deterministic and accepting.
The automaton
  additionally tracks the probability of the runs over time, to determine whether the
  individual runs and their sum have enough ``weight''. More precisely, it tracks the
  probabilities of the current prefixes, which in the limit yield the probabilities of the
  runs. As the runs we are interested in are limit-deterministic, there exists a finite
  prefix which has the probability of the whole run, hence tracking the prefix
  probabilities is sufficient for our purpose.

  The automaton rejects when the total probability of the guessed runs is
  $\leq \lambda$, one of the runs goes into the rejecting sink $q_{rej}$ or a run does not see
  accepting states infinitely often. Furthermore, the automaton shall guess no runs which
  are definitely useless for acceptance.
  By Lemma~\ref{lem:th_kpba_characterization} we only need to consider sets of runs with
  at most one run that has a probability $<\varepsilon$, where $\varepsilon :=
  \varepsilon_k$ is given by Lemma~\ref{lem:cutoff_fpba_th}. For this single run we also
  do not need to track the exact probability value, as its only purpose is to witness that the
  acceptance probability is strictly greater than $\lambda$, whereas all other runs must have
  one of the finitely many different probabilities which are $\geq \varepsilon$ and must
  sum to $\lambda$.

  Formally, let $\varepsilon$ be as in Lemma~\ref{lem:th_kpba_characterization}, and $V := V_{\geq \varepsilon} \cupdot \{\star_n, \star_d\}$, where
  $V_{\geq \varepsilon}$ is the finite (by Lemma~\ref{lem:run_finvals_above_th})
  set of different probability values $\geq \varepsilon$ that a run prefix of $\mc{A}$ can
  have, and the values $\star_n, \star_d$ are to be interpreted as arbitrarily small
  values such that $0 < \star_n, \star_d < \varepsilon$ and are introduced for
  convenience to cover the case of tracking a single low-probability run imprecisely.

  Then $\mc{B} := (Q', \Sigma, \Delta', Q_0', F_1,\ldots, F_k)$ is defined with
  \begin{itemize}
    \item $Q' := \bigcup_{i=1}^k (Q {\times} V)^i$ \hfill (tuples of at
      most $k$ states with probabilities),
    \item $Q_0' := \{ ((q_1,v_1)...(q_n,v_n)) \mid 1\leq n\leq k, q_i$ pw. diff. and $\forall (q_i,v_i), \mu_0(q_i)=v_i \}$,
    \item $F_i := \bigcup_{j=1}^{i-1}(Q {\times} V)^j
      \cup \{ ((q_1,v_1)...(q_i,v_i),...)\in Q' \mid q_i{\in}F,v_i{\neq}\star_n \}\ \forall i\in\{1...k\}$,
  \end{itemize}
  and for
      $S=((p_1,u_1),\ldots,(p_m,u_m)),T=((q_1,v_1),\ldots,(q_n,v_n)) \in Q'$ and symbol $a\in \Sigma$,
      the transition $(S,a,T)$ is defined in $\Delta'$ if
     \begin{itemize}
       \item $m \le n$, $\sum_{i=1}^n v_i > \lambda$ and $\forall i\in
         \{1\ldots n\}, q_i \neq q_{rej}$,
       \item there exists at most one $v_i$ such that $v_i < \varepsilon$, and
       \item there exist indices $1 = j_1 < \ldots < j_m \leq n$ and $j_{m+1}=n+1$ such that for
         all $i\in \{1\ldots m\}$:
         \begin{itemize}
           \item the states $q_{j_i}, \ldots, q_{j_{i+1}-1}$ are pairwise different, and
           \item for all $l \in \{j_i, \ldots, j_{i+1}-1\}$, we have:
             \begin{itemize}
               \item $v_l = u_i\cdot \delta(p_i,a,q_l)$ if $u_i\cdot \delta(p_i,a,q_l) \geq \varepsilon$,
               \item $v_l = \star_n$  if $u_i \geq \varepsilon$
                 and $0 < u_i\cdot \delta(p_i,a,q_l) < \varepsilon$,
               \item $v_l \in \{ \star_n, \star_d \}$  if $u_i = \star_n$ and $\delta(p_i,a,q_l)>0$
                 \hfill (guess when run is det.),
               \item $v_l = \star_d$  if $u_i = \star_d$ and $\delta(p_i,a,q_l) = 1$
                 \hfill (ensure that run det.).
             \end{itemize}
         \end{itemize}
     \end{itemize}

  This means that the automaton starts in a subset of the possible initial states (with
  respective initial probabilities), listed in a tuple in arbitrary order, and then must
  pick for each state at least one successor that has positive probability. For each state
  in the tuple also multiple different successors may be taken, which means that the
  automaton then tracks these as distinct runs, but the total number of tracked runs can
  be at most $k$. In other words, the automaton picks in each transition at most $k$
  different edges in the run tree of $\mc{A}$ and adjusts the probabilities according to
  the probability of the respective finite path prefix. Hence, by construction, the
  automaton tracks at most $k$ runs which are all different, all but at most one have a
  probability $\geq \varepsilon$,  no run ever goes into the rejecting sink $q_{rej}$ of
  $\mc{A}$, and the total probability of these runs is $>\lambda$.

  If $w\in L(\mc{B})$, then there exists an accepting run $\rho$ such that after some
  finite time $t$ the tuple size stabilizes at some size $n\leq k$ (as it is monotonically
  increasing) and the sum of probabilities in the tuple stabilizes at some value $>
  \lambda$ (as they are monotonically decreasing and can only take finitely many values).
  Furthermore,  as $\rho$ is accepting, infinitely many states along $\rho$ are in
  the sets $F_i$ for $i \leq n$, which means that in each tuple component accepting states
  of $\mc{A}$ are visited infinitely often. Notice that this implies that after $t$, every
  state in a tuple has exactly one selected successor, because each state must have at
  least one, but having more then one implies that the tuple would grow.
  Also, this successor must have probability 1 according to the transition distributions of
  $\mc{A}$, as either the total tracked probability would decrease, or there would be no
  transition for the single run which must at some point have the value $\star_d$ assigned.
  Hence, there exist $n$ different limit-deterministic accepting runs in $\mc{A}$ that
  have in total a probability $>\lambda$, witnessing that $w \in L^{>\lambda}(\mc{A})$.

  If $w \in L^{>\lambda}(\mc{A})$, then we can choose a set $R$ of accepting runs as in Lemma~\ref{lem:th_kpba_characterization}, i.e., with total probability $>\lambda$, at most one run with a
  probability $<\varepsilon$, and all subsets of $R$ have probability $< \lambda$.

The automaton $\mc{B}$ can guess this set $R$ of runs, increasing the size of the tuple whenever runs in $R$ separate after sharing a common prefix.
After some finite time then all those runs become deterministic, i.e.,
  only have unique successors with probability 1, which means that the tracked
  probabilities do not decrease anymore. For runs that have a probability $\geq
  \varepsilon$, this means that the tracked value stabilizes eventually. For the possible
  single run with probability $<\varepsilon$, the automaton eventually replaces its
  probability by $\star_n$ and finally by $\star_d$, after the run has also become
  deterministic. As by assumption the runs are accepting, in every component of the tuple
  infinitely often an accepting state is visited, such that by definition, infinitely
  often a state in $F_i$ is visited for all $1\leq i \leq k$, hence $w\in L(\mc{B})$.
  \lncsqed

\end{proof}

 \clearpage
\section{Proofs for section on weak PBA}
\label{app:weak}

\subsection{Proof for Theorem~\ref{thm:reg_undec}}
\thmregundec*
\begin{proof}
    Since $\mbb{L}^{>\lambda}(\text{PWA}) \supseteq
  \mbb{L}^{>0}(\text{PWA})$ (see Theorem~\ref{thm:expressive_power_pwa}), 
  $\mbb{L}^{>0}(\text{PWA}) = \overline{\mbb{L}^{=1}(\text{PWA})}$
  (see remark above),
  and the class of regular $\omega$-languages is closed under
  complement, it suffices to show the statement for PWA$^{=1}$.
We do this by reduction from the value 1
  problem for PFA, which is the question whether for each $\varepsilon>0$ there exists a
  word accepted by the PFA with probability $> 1-\varepsilon$. This problem is known to be
  undecidable \cite{gimbert2010probabilistic}. We consider a
  slightly modified version of the problem by assuming that no word
  is accepted with probability 1 by the given PFA. The problem remains
  undecidable under this assumption, because one can check if a
  PFA accepts a finite word with probability 1 by a simple subset construction.

  Let $\mc{A}=(Q,\Sigma,\delta,\mu_0,F)$ be some PFA. We construct a PWA $\mc{B}$ by
  taking a copy of $\mc{A}$ and extending it with a new symbol $\#$ such that from
  accepting states of $\mc{A}$ the automaton is ``restarted'' on $\#$, while from
  non-accepting states $\#$ leads into a new part which ensures that infinitely many
  $\#$ are seen and contains the only accepting state of $\mc{B}$.

  Formally, we construct the PWA $\mc{B}=(Q', \Sigma',\delta',\mu_0,F')$ with
  $Q' := Q\cupdot\{q_\#, q_a\}, \Sigma':=\Sigma\cupdot\{\#\}$, $F' := \{q_a\}$ by
  extending $\delta$ to $\delta'$ as follows:
  \begin{itemize}
    \item $\delta'(p,x,q) := \delta(p,x,q)$ $\forall p,q\in Q, x\in \Sigma$,
    \item $\delta'(p,\#,q) := \mu_0(q)$ if $p\in F$ and
    $\delta'(p,\#,q_\#) = 1$ if $p\in Q\setminus F$,
    \item $\delta'(q_a,\#,q_a) = \delta'(q_a,x,q_a) = \delta'(q_\#,x,q_\#) = 1\  \forall x\in\Sigma$, and
    \item $\delta'(q_\#,\#,q_\#) = \delta'(q_\#,\#,q_a) = \frac{1}{2}$.
  \end{itemize}

  First notice that whenever a run reaches $q_\#$, its continuations will almost surely
  reach $q_a$ (and hence be accepting) iff $\#$ is read infinitely often.

  If $\mc{A}$ does not have value 1, then there exists some $\varepsilon>0$ such that
  every word is accepted by $\mc{A}$ with probability $\leq 1-\varepsilon$.
  But as $q_\#$ can only be avoided by reaching a state that is accepting in $\mc{A}$
  before reading $\#$, for any infinite sequence of words $w_i \in \Sigma^*$ for $i\in
  \mbb{N}$ we have that the probability to never reach $q_\#$ on the word
  $w=w_1\#w_2\#\ldots$ is $\prod_i \op{Pr}_{Acc}(\mc{A},w_i) \leq \prod_i 1-\varepsilon =
  0$, which means that on any such $w$ almost surely the state $q_\#$ will be reached and
  hence $w$ will be accepted. For words not of this shape, i.e. containing only finitely
  many $\#$, a run will either never reach $q_\#$ or stay in it forever never reaching
  $q_a$. Therefore we have $L^{=1}(\mc{B}) = (\Sigma^*\#)^\omega$, which is a
  regular language.

  For the case that $\mc{A}$ does have value 1, recall that we
  assumed that no word is accepted with probability 1.
  But since there are words accepted with probability arbitrarily
  close to 1, there exists an infinite sequence of words $w_i \in \Sigma^*$ such that $\prod_i
  \op{Pr}_{Acc}(\mc{A},w_i) > 0$, and therefore on $w=w_1\#w_2\#\ldots$ with positive
  probability $q_\#$ can be avoided forever, i.e., $w \not\in L^{=1}(\mc{B})$. Notice that
  such a word $w$ cannot be ultimately periodic, as then $w$ could be written as
  $uv^\omega$ where $v=w_j\#w_{j+1}\#\ldots w_k\#$ for some $j,k \in \mbb{N}, j\leq k$. If $p$ is
  the probability to avoid $q_\#$ on $v$ in $\mc{B}$, then the probability to avoid $q_\#$ on
  $w$ is at most $\prod_i p$, which is 0 for $p<1$ and we already excluded that $p=1$
  (this would require that at least one word is accepted by $\mc{A}$ with probability 1),
  so all ultimately periodic words are accepted by $\mc{B}$. But then the subset
  $R\subseteq (\Sigma^*\#)^\omega$ of words of the shape $w_1\#w_2\#\ldots$ that are
  rejected by $\mc{B}$ does not contain an ultimately periodic word, so $R$ cannot
  be regular and therefore $L^{=1}(\mc{B}) = (\Sigma^*\#)^\omega \setminus R$ is
  also not regular.
  \lncsqed
\end{proof}

\subsection{Proof for Proposition~\ref{prp:th_pwa_nonreg_strong}}
\vspace{-5mm}
\begin{figure}[h!]
  \centering
  \begin{tikzpicture}[baseline={([yshift=-.5ex]current bounding box.center)},shorten >=1pt,
    node distance=1.2cm,inner sep=1pt,on grid,auto]
    \node[state,initial, initial text={$\frac{1}{2}$}] (qa)   {$q_a$};
    \node[state,initial, initial text={$\frac{1}{2}$},below=of qa] (qb)   {$q_b$};
    \node[state,right=of qa] (qp) {$q_+$};
    \node[state,accepting,right=of qb] (qs) {$q_\$$};
      \path[->]
      (qa) edge [loop above] node[xshift=10mm,yshift=-3mm] {$b:1,a:\frac{1}{2}$} (qa)
      (qa) edge [swap] node[yshift=-2pt] {$a:\frac{1}{2}$} (qp)
      (qp) edge [loop right] node {$a,b$} (qp)
      (qb) edge [loop below] node[xshift=10mm,yshift=3mm] {$a:1,b:\frac{1}{2}$} (qb)
      (qb) edge [] node {\$} (qs)
      (qp) edge [] node {\$} (qs)
      (qs) edge [loop right] node {\$} (qs)
      ;
  \end{tikzpicture}
  \captionsetup{labelformat=empty}
  \caption{Automaton in Figure~\ref{fig:pwa_fig}(a).}
\end{figure}
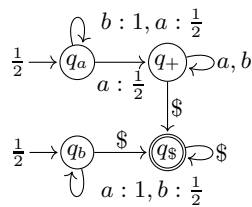\addtocounter{figure}{-1}
\vspace{-5mm}

\prpthpwanonreg*
\begin{proof}
  We show the result for $\lambda=\frac{1}{2}$ (in which case the PWA even has only
  rational coefficients). The general statement follows, because
  one can easily modify the PBA to accept the same language with any threshold $\lambda\in ]0,1[$
  by \cite[Lemma 4.15]{baier2012probabilistic}.

  Consider the PWA $\mc{A}$ in Figure~\ref{fig:pwa_fig}(a). Clearly,
  it can only positively accept words of shape $(a+b)^*\$^\omega$.
  Let $w=u\$^\omega$ with $u\in \{a,b\}^*$ and let $\#_a(u)$ denote the number of occurrences of
  $a\in \Sigma$ in $u$.
  Notice that on each $b$, half of the remaining probability of currently being in
  $q_b$ goes into the (implicit)
  rejecting sink, and on each $a$, half the probability
  of currently being at $q_a$ goes to $q_+$.
  The only runs which can continue on $\$^\omega$ after reading $u$ are
  in $q_b$ or in $q_+$ after $u$ and the unique possible run continuation on $\$^\omega$
  goes to and forever stays in the accepting state $q_\$$. Hence, we have:

  $\op{Pr}(\mc{A}$ accepts $w$) = $\overbrace{\op{Pr}(\mc{A}\text{ \ in\ } q_b \text{ after }u)}^{
      \frac{1}{2}\cdot \frac{1}{2}^{\#_b(u)}}$\quad
    + $\overbrace{\op{Pr}(\mc{A}\text{\ in\ } q_{+} \text{ after } u)}^{
      \frac{1}{2}\sum_{i=1}^{\#_a(u)} \frac{1}{2}^i = \frac{1}{2}\cdot (1 - \frac{1}{2}^{\#_a(u)})}$

  This means, that $\op{Pr}(\mc{A}$ accepts $w$) =
  $\frac{1}{2}\cdot(1-\frac{1}{2}^{\#_a(u)} + \frac{1}{2}^{\#_b(u)})$, which is greater
  than $\frac{1}{2}$ if and only if $\#_a(u) > \#_b(u)$, and therefore
  $L^{>\frac{1}{2}}(\mc{A}) = \{\ (a+b)^*\$^\omega \mid \#_a(u) > \#_b(u)\ \}$.

  Now it is easy to see that there are infinitely many Myhill-Nerode equivalence classes
  for this language, and hence it cannot be regular (as the implication
  ``regular $\Rightarrow$ finitely many Myhill-Nerode classes'' also holds for infinite
  words). Furthermore, by Lemma~\ref{lem:fin_myhnerode_classes} languages accepted by
  PBA$^{>0}$ have only finitely many classes. Hence, this language cannot be accepted by
  any PBA$^{>0}$.
  \lncsqed
\end{proof}

\clearpage

\subsection{Proof for Theorem~\ref{thm:pca_to_pwa}}
\thmpcatopwa*
\begin{proof}
  We show the first statement. The second then follows by duality, i.e., we can
  interpret a PBA$^{=1}$ $\mc{A}$ recognizing $L$ as a PCA$^{>0}$ recognizing
  $\overline{L}$ and just apply the construction to get a PWA$^{>0}$ $\mc{B}$ for $\overline{L}$,
  such that $\overline{\mc{B}}$ (with inverted accepting and rejecting states) is a
  PWA$^{=1}$ for $L$. In the first statement the $\subseteq$ inclusion is trivial, hence
  we only need to show that $\mathbb{L}^{>0}(\PCoBA) \subseteq \mathbb{L}^{>0}(\PWA)$.

  Now let $\mc{A}=(Q,\Sigma,\delta,\mu_0,F)$ be a PCA$^{>0}$. We refer to the states in $F$
  as \emph{bad} states (since they occur only finitely often in accepting runs). Intuitively, the
  PWA$^{>0}$ $\mc{B}$ accepting the same language is constructed as follows.
  Take two copies of $\mc{A}$, a \emph{guess} copy and a \emph{verify} copy.
  Each transition in the guess copy is modified to go into the verify copy with probability
  $\frac{1}{2}$ and all transitions to copies of bad states in the verify copy are
  redirected to a rejecting sink.

  Formally, let $Q_g$, $Q_v$ be two copies of the states $Q$ and let $q^g$ and $q^v$
  denote the respective copy of $q\in Q$. The PWA $\mc{B} = (Q', \Sigma, \delta', \mu_0',
  F')$ is defined with $Q' := Q_g \cup Q_v\cup\{q_{rej}\}, \mu_0'(q^g) := \mu_0(q)$ for all
  $q^g \in Q_g$ and $0$ otherwise, $F' := Q_v$, and $\delta'$ defined as:
  \begin{itemize}
    \item $\delta'(p^g, x, q^g) = \delta'(p^g, x, q^v) = \frac{1}{2}\cdot \delta(p,x,q)$
    \item $\delta'(p^v, x, q^v) = \delta(p,x,q)$ \quad if $q\not\in F$
    \item $\delta'(p^v, x, q_{rej}) = 1 - \sum_{q \not\in F} \delta(p,x,q)$
  \end{itemize}

  Notice that we can write the set of accepting runs $\op{AccRuns}(\mc{A},w)$ on some word
  $w\in \Sigma^\omega$ as a countable union of disjoint sets $\bigcup_{i\geq 0}
  \op{goodFrom(i)}$, such that $\op{goodFrom(i)}$ contains the accepting runs where $i$ is
  the smallest time such that no state in $F$ is visited at times $\geq i$.

  Assume that $w \in L^{>0}(\mc{A})$. By $\sigma$-additivity, this implies
  $\op{Pr}(\op{AccRuns}(\mc{A},w)) = \sum_{i\geq 0} \op{Pr}(\op{goodFrom}(i)) > 0$ and
  hence there is an $i$ with $\op{Pr}(\op{goodFrom}(i)) > 0$. Let $Q_i \subseteq Q$
  be the set of states occupied by some run in $\op{goodFrom}(i)$ at time $i$. Clearly $Q_i$ is
  reached at time $i$ with positive probability and by definition the
  runs in $\op{goodFrom}(i)$ never see bad states after $i$. But then by construction,
  with positive probability some runs of $\mc{B}$ stay in the guess copy until time $i-1$ and
  reach the verify copy at time $i$ and then they proceed in the verify copy exactly as the
  runs $\op{goodFrom}(i)$ proceed after $i$ in $\mc{A}$. Hence, they never will visit states $q^v$
  which correspond to states $q\in F$ and thus forever stay in the verify copy (where all
  states are accepting) and therefore $w \in L^{>0}(\mc{B})$.

  The other direction is similar---if $w\in L^{>0}(\mc{B})$, then
  there exists some time $i$ such that runs of $\mc{B}$ reach the verify copy at $i$ and
  then with positive probability stay there, i.e., there is a subset $\op{goodFrom}(i)$ of
  those runs that has positive probability, such that the runs never visit the
  rejecting sink after reaching $i$. By construction, clearly the probability for
  corresponding runs in $\mc{A}$ is at least as large and hence $w \in L^{>0}(\mc{A})$.
  \lncsqed
\end{proof}

\clearpage

\subsection{Proof details for Theorem~\ref{thm:expressive_power_pwa}(2) }

In this section we show that $\mbb{L}^{>0}(\PBA)$ and $\mbb{L}^{>\lambda}(\PWA)$ are
incomparable, i.e., neither contains the other one.
One direction directly follows by Proposition~\ref{prp:th_pwa_nonreg_strong}, i.e., there are
languages recognized by PWA$^{>\lambda}$ that cannot be recognized with PBA$^{>0}$.

For the other direction, the following result characterizes the languages accepted
by weak automata under extremal semantics in the Borel hierarchy, from which the claim
will follow.
We do not introduce the details of this hierarchy here, but rather refer the reader not
familiar with these concepts to \cite{staiger1983finite} and
\cite{chadha2011randomization}. Notice that the sets we call $\Pi_2$ and $\Sigma_2$ (using
modern naming) are called $G_\delta$ and $F_\sigma$ there.

The result easily follows from an adaptation of \cite[Lemma 3.2]{chadha2011randomization}:
\begin{lemma}[Topological characterization]
  \label{lem:pwa_topo}
  If $\mc{A}$ is a PWA and $\lambda\in [0,1]$ a threshold, then
  $L^{\geq \lambda}(\mc{A})$ is a $\Pi_2$ set and $L^{> \lambda}(\mc{A})$ is a $\Sigma_2$ set.
\end{lemma}
\begin{proof}
  The first statement is implied by \cite[Lemma 3.2]{chadha2011randomization}, as
  $L^{\geq \lambda}(\mc{A})$ is a $\Pi_2$ set for any (even not weak) PBA. The second
  statement can be obtained for weak automata by a simple adaptation of this proof, by
  showing that the set of words rejected by some PWA with probability $\leq (1-\lambda)$ is
  a $\Pi_2$ set. The decomposition of paths into countable unions and intersections
  performed in the proof can be done in the same way, due to the fact that in weak
  automata a run is rejecting if it sees rejecting states infinitely often (which means
  that the run eventually stays in a rejecting SCC). But then clearly the complement of
  this set is the set of words that are accepted by $\mc{A}$ with probability $> \lambda$,
  which is exactly $L^{>\lambda}(\mc{A})$ and by definition is a $\Sigma_2$ set.
  \qed
\end{proof}

From Lemma~\ref{lem:pwa_topo} and the facts shown in \cite{chadha2011randomization} that
$\mbb{L}^{>0}(\PBA) = \op{BCl}(\mbb{L}^{=1}(\PBA))$ and $\mbb{L}^{=1}(\PBA) \subseteq \Pi_2$, we
conclude that PBA$^{>0}$ especially can recognize some languages in $\Pi_2$, whereas
PWA$^{>\lambda}$ can only recognize languages in $\Sigma_2$.

\vfill
{\small\medskip\noindent{\bf Open Access} This chapter is licensed under the terms of the Creative Commons\break Attribution 4.0 International License (\url{http://creativecommons.org/licenses/by/4.0/}), which permits use, sharing, adaptation, distribution and reproduction in any medium or format, as long as you give appropriate credit to the original author(s) and the source, provide a link to the Creative Commons license and indicate if changes were made.}
{\small \spaceskip .28em plus .1em minus .1em The images or other third party material in this chapter are included in the chapter's Creative Commons license, unless indicated otherwise in a credit line to the material.~If material is not included in the chapter's Creative Commons license and your intended\break use is not permitted by statutory regulation or exceeds the permitted use, you will need to obtain permission directly from the copyright holder.}
\medskip\noindent\includegraphics{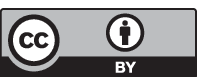}

\end{document}